\documentclass{amsart}%
\usepackage{amsfonts}
\usepackage{amsmath}
\usepackage{amssymb}
\usepackage{graphicx}
\usepackage[latin1]{inputenc}
\usepackage[english]{babel}%
\newtheorem{theorem}{Theorem}[section]
\theoremstyle{plain}
\newtheorem{acknowledgement}[theorem]{Acknowledgement}

\newtheorem{corollary}[theorem]{Corollary}

\newtheorem{definition}[theorem]{Definition}
\newtheorem{example}[theorem]{Example}

\newtheorem{lemma}[theorem]{Lemma}

\newtheorem{proposition}[theorem]{Proposition}
\newtheorem{remark}[theorem]{Remark}

\numberwithin{equation}{section}
\begin{document}
\title[Non-Archimedean Euclidean Fields]{Non-Archimedean White Noise, Pseudodifferential Stochastic Equations, and
Massive Euclidean Fields.}
\author{W. A. Zúñiga-Galindo}
\address{Centro de Investigación y de Estudios Avanzados del Instituto Politécnico Nacional\\
Departamento de Matemáticas, Unidad Querétaro\\
Libramiento Norponiente \#2000, Fracc. Real de Juriquilla. Santiago de
Querétaro, Qro. 76230\\
México.}
\email{wazuniga@math.cinvestav.edu.mx}
\thanks{}

\begin{abstract}
We construct $p$-adic Euclidean random fields $\boldsymbol{\Phi}$ over
$\mathbb{Q}_{p}^{N}$, for arbitrary $N$, \ these fields are solutions of
$p$-adic stochastic pseudodifferential equations. From a mathematical
perspective, the Euclidean fields are generalized stochastic processes
parametrized by functions belonging to a nuclear countably Hilbert space,
these spaces are introduced in this article, in addition, the Euclidean fields
are invariant under the action of certain group of transformations. We also
study the Schwinger functions of $\boldsymbol{\Phi}$.

\end{abstract}
\maketitle

\section{Introduction}

There are general arguments that suggest that one cannot make measurements in
regions of extent smaller than the Planck length $\approx10^{-33}$ cm, see
e.g. \cite[and the references therein]{Var1}. The construction of physical
models at the level of Planck scale is a relevant scientific problem and a
very important source for mathematical research. In \cite{Volovich1}%
-\cite{Volovich2}, I. Volovich posed the conjecture of the non-Archimedean
nature of the spacetime at the level of the Planck scale. This conjecture has
originated a lot research, for instance, in quantum mechanics, see e.g.
\cite{kh0}, \cite{Vlad-Vol}, \cite{V-V-Z-1}, \cite{Zel1}, \cite{Zel2},\ in
string theory, see e.g. \cite{Bre-fre-ols-wi}, \cite{Fran}, \cite{G-S},
\cite{M-Z}, \cite{Vlad-Vol2}, \cite{Vlad1}, \cite{Vlad2}, and in quantum field
theory, see e.g. \cite{Koch-Sait}, \cite{Mis}, \cite{Smir}. For a further
discussion on non-Archimedean mathematical physics, the reader may consult
\cite{D-K-K-V}, \cite{Var2}, \cite{V-V-Z} and the references therein. On the
other hand, the interaction between quantum field theory and mathematics is
very fruitful and deep, see e.g. \cite{Frenkel}, \cite{Gaisgory},
\cite{Kas-Wi}, \cite{kh1}-\cite{kh2}, \cite{Wi-1}-\cite{Wi-2}, among several
articles. Let us mention explicitly the connection with arithmetic, see e.g.
\cite{Kas-Wi}, \cite{Lax-Phillips}, \cite{Palov-Fadeev}, from this perspective
the investigation of quantum fields in the non-Archimedean setting is a quite
natural problem.

In this article we introduce a class of non-Archimedean Euclidean fields, in
arbitrary dimension, which are constructed as solutions of certain covariant
$p$-adic stochastic pseudodifferential equations, by using techniques of white
noise calculus. The connection between quantum fields and SPDEs has been
studied intensively in the Archimedean setting, see e.g.
\cite{Albeverio-et-al-1}-\cite{Albeverio-et-al-4} and the references therein.
A massive non-Archimedean field $\boldsymbol{\Phi}$ is a random field
parametrized by a nuclear countably Hilbert space $\mathcal{H}_{\mathbb{R}%
}\left(  \mathbb{Q}_{p}^{N};\alpha,\infty\right)  $, \ which depend on
$\left(  \mathfrak{q},\mathfrak{l},m,\alpha\right)  $, where $\mathfrak{q}$ is
an elliptic quadratic form, $\mathfrak{l}$ is an elliptic polynomial, and $m$,
$\alpha$ are positive numbers, here $m$ is the mass of $\boldsymbol{\Phi}$.
Heuristically, $\boldsymbol{\Phi}$ is the solution of $\left(  \boldsymbol{L}%
_{\alpha}+m^{2}\right)  \boldsymbol{\Phi}=\boldsymbol{\digamma}$, where
$\digamma$ is a generalized Lévy noise, this type of noise is introduced in
this article, and $\mathbf{L}_{\alpha}\left(  \cdot\right)  =\mathcal{F}%
_{\xi\rightarrow x}^{-1}\left(  \left\vert \mathfrak{l}\left(  \xi\right)
\right\vert _{p}^{\alpha}\mathcal{F}_{x\rightarrow\xi}\left(  \cdot\right)
\right)  $, where $\mathcal{F}$ is the Fourier transform on $\mathbb{Q}%
_{p}^{N}$ defined using the bilinear symmetric form corresponding to the
quadratic form $\mathfrak{q}$. The operator $\boldsymbol{L}_{\alpha}+m^{2}$ is
a non-Archimedean analog of a fractional Klein-Gordon operator. At this point,
it is useful to compare our construction with the classical one. In the
Archimedean case, see e.g. \cite{Albeverio-et-al-3}-\cite{Albeverio-et-al-4},
the elliptic quadratic form $\mathfrak{b}(\xi)=\xi_{1}^{2}+\cdots+\xi_{N}%
^{2}\in\mathbb{R}\left[  \xi_{1},\cdots,\xi_{N}\right]  $ is used to define
the pseudodifferential fractional Klein-Gordon operator $\left(  -\Delta
+m^{2}\right)  ^{\alpha}$, $m$, $\alpha>0$, and also $\mathfrak{b}(\xi)$ is
the quadratic form associated with the bilinear form $\xi_{1}\eta_{1}%
+\cdots+\xi_{N}\eta_{N}$, which is used in the definition of the Fourier
transform on $\mathbb{R}^{N}$. This approach cannot be carried out in the
$p$-adic setting because when $\mathfrak{b}(\xi)$ is considered as a $p$-adic
quadratic form, it is not elliptic for $N\geq5$, and \ the ellipticity of
$\mathfrak{b}(\xi)$ is essential to establish the non-negativity of the Green
functions in the \ $p$-adic setting. For this reason, we have to replace
$\xi_{1}^{2}+\cdots+\xi_{N}^{2}$ by an elliptic polynomial, which is a
homogeneous polynomial that vanishes only at the origin, there are infinitely
many of them. The symmetries of our Green functions and fields depend on the
transformations that preserve $\mathfrak{q}$ and $\mathfrak{l}$, thus, in
order to have large groups of symmetries we cannot fix $\mathfrak{q}$,
instead, we work with pairs $\left(  \mathfrak{q},\mathfrak{l}\right)  $ that
have a large group of symmetries. These are two important difference between
the $p$-adic case and the Archimedean one. On the other hand, the spaces
$\mathcal{H}_{\mathbb{R}}\left(  \mathbb{Q}_{p}^{N};\alpha,\infty\right)  $
introduced here are completely necessary to carry out a construction similar
to the one presented in \cite{Albeverio-et-al-4}, \cite{Albeverio-et-al-3},
For instance, the Green function attached to $\boldsymbol{L}_{\alpha}+m^{2}$
gives rise to a continuous mapping from $\mathcal{H}_{\mathbb{R}}\left(
\mathbb{Q}_{p}^{N};\alpha,\infty\right)  $\ into itself. This fact is not
true, if we replace $\mathcal{H}_{\mathbb{R}}\left(  \mathbb{Q}_{p}^{N}%
;\alpha,\infty\right)  $ by the space of test functions $\mathcal{D}%
_{\mathbb{R}}\left(  \mathbb{Q}_{p}^{N}\right)  $, which is also nuclear. We
want to mention here that \cite{Albeverio-et-al-4}, \cite{Albeverio-et-al-3}
have influenced strongly our work, and that our results are non-Archimedean
counterparts of some of the results of \cite{Albeverio-et-al-3}.

The line of research pursued in this article continues, from one side, the
investigations on generalized stochastic processes invariant under groups of
transformations, see e.g. \cite{Kolmogorov}, \cite{Minlos}, \cite{Gel-Vil}
and, on the other side, it continues the development of the white noise
calculus, see e.g. \cite{Hida et al}, \cite{Obata}, in the non-Archimedean
setting. Several types of $p$-adic noise calculi\ have been studied lately,
see e.g. \cite{Bikulov}, \cite{Bi-Volovich}, \cite{Evans}, \cite{Kamizono},
\cite{K-Kos1}, \cite{K-Kos2}, \cite{K-Zh-1}, \cite{K-Zh-2}, \cite{Zu2}. To the
best of our knowledge the non- Archimedean generalized Lévy noises studied
here have not been reported before in the literature.

Finally, we want to compare our results with those of \cite{Koch-Sait}\ and
\cite{Abd et al}. In \cite{Koch-Sait} the authors constructed non-Gaussian
measures on the space $\mathcal{D}(\mathbb{Q}_{p}^{n})$ of distributions on
$\mathbb{Q}_{p}^{n}$, for $1\leq n\leq4$. The measures are absolutely
continuous with respect to the free field Gaussian measure and live on the
spaces of distributions with a fixed compact support (finite volume case). As
we mention before our results cannot be formulated in the space of test
functions $\mathcal{D}(\mathbb{Q}_{p}^{n})$, because it is nonmetrizable
topological space, and thus, the theory of countably Hilbert nuclear spaces
does not apply to it. The author thinks that the results of \cite{Koch-Sait}
can be extended by using our spaces $\mathcal{H}_{\mathbb{R}}\left(
\mathbb{Q}_{p}^{N};\alpha,\infty\right)  $, but more importantly, the author
thinks that the problem of infinite volume limit in \cite{Koch-Sait} can be
settled by using the results presented here in combination with those of
\cite{Albeverio-et-al-4}. In \cite{Abd et al}, three dimensional $p$-adic
massless quantum fields were constructed and investigated in connection with
the renormalization group. This article motivates the study of spaces of type
$\mathcal{H}_{\mathbb{R}}\left(  \mathbb{Q}_{p}^{N};\alpha,\infty\right)  $ in
the contest of massless quantum fields. The author expect to consider these
matters and some related in forthcoming articles. At the moment, we do not
know the exact physical meaning of the quantum fields constructed here neither
their connections with the non-Archimedean $\varphi^{4}$ models developed in
\cite{Mis1}-\cite{Mis2}. We hope that this article motivates the study of
these matters.

The article is organized as follows. In Section \ref{Sect2}, we review the
basic aspects of the $p$-adic analysis. In Section \ref{Sect3}, we construct
the spaces $\mathcal{H}_{\mathbb{R}}\left(  \mathbb{Q}_{p}^{N};\alpha
,\infty\right)  $, $\mathcal{H}_{\mathbb{C}}\left(  \mathbb{Q}_{p}^{N}%
;\alpha,\infty\right)  $, which are nuclear countably Hilbert spaces, see
Theorem \ref{Thm1} and Remark \ref{nota1}. We also show that $\mathcal{H}%
_{\mathbb{C}}\left(  \mathbb{Q}_{p}^{N};\alpha,\infty\right)  $\ is invariant
under the action of a large class of pseudodifferential operators, see Theorem
\ref{Thm2}. In Section \ref{Sect4}, we introduce non-Archimedean analogs of
the Klein-Gordon fractional operators and study the properties of the
corresponding Green functions, see Proposition \ref{Prop1}. We also study the
solutions of the $p$-adic Klein-Gordon equations in $\mathcal{H}_{\mathbb{R}%
}\left(  \mathbb{Q}_{p}^{N};\alpha,\infty\right)  $, see Theorem \ref{Thm3}.
In Section \ref{Sect5}, we introduce a new class of non-Archimedean Lévy
noises, see Theorem \ref{Thm5} and Definition \ref{def2}. Section \ref{Sect6}
is dedicated to the non-Archimedean quantum fields and its symmetries, see
Proposition \ref{Prop3}, Definitions \ref{def_6A}-\ref{def_6} and Proposition
\ref{Prop4}. Finally, Section \ref{Sect7} is dedicated to the study of the
Schwinger functions, see Definition \ref{def6}\ and Theorem \ref{Thm6}.
Finally, as an application,we \ construct a $p$-adic Brownian sheet on
$\mathbb{Q}_{p}^{N}$, see Theorem \ref{Thm7}.

\section{\label{Sect2}$p$-adic analysis: essential ideas}

In this section we fix the notation and collect some basic results on $p$-adic
analysis that we will use through the article. For a detailed exposition the
reader may consult \cite{A-K-S}, \cite{Taibleson}, \cite{V-V-Z}.

\subsection{The field of $p$-adic numbers}

Along this article $p$ will denote a prime number. The field of $p-$adic
numbers $\mathbb{Q}_{p}$ is defined as the completion of the field of rational
numbers $\mathbb{Q}$ with respect to the $p-$adic norm $|\cdot|_{p}$, which is
defined as
\[
|x|_{p}=%
\begin{cases}
0 & \text{if }x=0\\
p^{-\gamma} & \text{if }x=p^{\gamma}\dfrac{a}{b},
\end{cases}
\]
where $a$ and $b$ are integers coprime with $p$. The integer $\gamma:=ord(x)$,
with $ord(0):=+\infty$, is called the\textit{ }$p-$\textit{adic order of} $x$.
We extend the $p-$adic norm to $\mathbb{Q}_{p}^{N}$ by taking%
\[
||x||_{p}:=\max_{1\leq i\leq N}|x_{i}|_{p},\qquad\text{for }x=(x_{1}%
,\dots,x_{N})\in\mathbb{Q}_{p}^{N}.
\]

We define $ord(x)=\min_{1\leq i\leq N}\{ord(x_{i})\}$, then $||x||_{p}%
=p^{-ord(x)}$. The metric space $\left(  \mathbb{Q}_{p}^{N},||\cdot
||_{p}\right)  $ is a complete ultrametric space. As a topological space
$\mathbb{Q}_{p}$\ is homeomorphic to a Cantor-like subset of the real line,
see e.g. \cite{A-K-S}, \cite{V-V-Z}.

Any $p-$adic number $x\neq0$ has a unique expansion $x=p^{ord(x)}\sum
_{j=0}^{+\infty}x_{j}p^{j}$, where $x_{j}\in\{0,1,2,\dots,p-1\}$ and
$x_{0}\neq0$. By using this expansion, we define \textit{the fractional part
of }$x\in\mathbb{Q}_{p}$, denoted $\{x\}_{p}$, as the rational number
\[
\{x\}_{p}=%
\begin{cases}
0 & \text{if }x=0\text{ or }ord(x)\geq0\\
p^{ord(x)}\sum_{j=0}^{-ord(x)-1}x_{j}p^{j} & \text{if }ord(x)<0.
\end{cases}
\]
For $l\in\mathbb{Z}$, denote by $B_{l}^{N}(a)=\{x\in\mathbb{Q}_{p}%
^{N}:||x-a||_{p}\leq p^{l}\}$ \textit{the ball of radius }$p^{l}$ \textit{with
center at} $a=(a_{1},\dots,a_{N})\in\mathbb{Q}_{p}^{N}$, and take $B_{l}%
^{N}(0):=B_{l}^{N}$. Note that $B_{l}^{N}(a)=B_{l}(a_{1})\times\cdots\times
B_{l}(a_{n})$, where $B_{l}(a_{i}):=\{x\in\mathbb{Q}_{p}:|x-a_{i}|_{p}\leq
p^{l}\}$ is the one-dimensional ball of radius $p^{l}$ with center at
$a_{i}\in\mathbb{Q}_{p}$. The ball $B_{0}^{N}$ equals the product of $N$
copies of $B_{0}:=\mathbb{Z}_{p}$, \textit{the ring of }$p-$\textit{adic
integers}. We denote by $\Omega(\left\Vert x\right\Vert _{p})$ the
characteristic function of $B_{0}^{N}$. For more general sets, say Borel sets,
we use ${\LARGE 1}_{A}\left(  x\right)  $ to denote the characteristic
function of $A$. For $l\in\mathbb{Z}$, denote by $S_{l}^{N}(a)=\{x\in
\mathbb{Q}_{p}^{N}:||x-a||_{p}=p^{l}\}$ \textit{the sphere of radius }$p^{l}$
\textit{with center at} $a=(a_{1},\dots,a_{N})\in\mathbb{Q}_{p}^{N}$, and take
$S_{l}^{N}(0):=S_{l}^{N}$.

\subsection{The Bruhat-Schwartz space}

A complex-valued function $g$ defined on $\mathbb{Q}_{p}^{N}$ is
\textit{called locally constant} if for any $x\in\mathbb{Q}_{p}^{N}$ there
exists an integer $l(x)\in\mathbb{Z}$ such that%
\begin{equation}
g(x+x^{\prime})=g(x)\text{ for }x^{\prime}\in B_{l(x)}^{N}.
\label{local_constancy}%
\end{equation}
A function $g:\mathbb{Q}_{p}^{N}\rightarrow\mathbb{C}$ is called a
\textit{Bruhat-Schwartz function (or a test function)} if it is locally
constant with compact support. The $\mathbb{C}$-vector space of
Bruhat-Schwartz functions is denoted by $\mathcal{D}_{\mathbb{C}}%
(\mathbb{Q}_{p}^{N})$, while $\mathcal{D}_{\mathbb{R}}(\mathbb{Q}_{p}^{N}%
)$\ denotes the $\mathbb{R}$-vector space of Bruhat-Schwartz functions.

Let $\mathcal{D}_{\mathbb{C}}^{^{\ast}}(\mathbb{Q}_{p}^{N})$ (resp.
$\mathcal{D}_{\mathbb{R}}^{^{\ast}}(\mathbb{Q}_{p}^{N})$) denote the set of
all continuous functionals (distributions) on $\mathcal{D}_{\mathbb{C}%
}(\mathbb{Q}_{p}^{N})$ (resp. $\mathcal{D}_{\mathbb{R}}(\mathbb{Q}_{p}^{N})$).

\subsection{\label{Sect_Fourier_Trans}Fourier transform}

Set $\chi_{p}(y)=\exp(2\pi i\{y\}_{p})$ for $y\in\mathbb{Q}_{p}$. The map
$\chi_{p}(\cdot)$ is an additive character on $\mathbb{Q}_{p}$, i.e. a
continuos map from $\mathbb{Q}_{p}$ into the unit circle satisfying $\chi
_{p}(y_{0}+y_{1})=\chi_{p}(y_{0})\chi_{p}(y_{1})$, $y_{0},y_{1}\in
\mathbb{Q}_{p}$.

Let $\mathfrak{B}\left(  x,y\right)  $\ be a symmetric non-degenerate
$\mathbb{Q}_{p}-$bilinear form on $\mathbb{Q}_{p}^{N}\times\mathbb{Q}_{p}^{N}%
$. Thus $\mathfrak{q}(x):=\mathfrak{B}\left(  x,x\right)  $, $x\in
\mathbb{Q}_{p}^{N}$ is a \textit{non-degenerate quadratic form} on
$\mathbb{Q}_{p}^{N}$.\ We recall that
\begin{equation}
\mathfrak{B}\left(  x,y\right)  =\frac{1}{2}\left\{  \mathfrak{q}%
(x+y)-\mathfrak{q}(x)-\mathfrak{q}(y)\right\}  . \label{Bilinear}%
\end{equation}
We identify the $\mathbb{Q}_{p}$-vector space $\mathbb{Q}_{p}^{N}$ with its
algebraic dual $\left(  \mathbb{Q}_{p}^{N}\right)  ^{\ast}$ by means of
$\mathfrak{B}\left(  \cdot,\cdot\right)  $. We now identify the dual group
(i.e. the Pontryagin dual) of $\left(  \mathbb{Q}_{p}^{N},+\right)  $ with
$\left(  \mathbb{Q}_{p}^{N}\right)  ^{\ast}$ by taking $x^{\ast}\left(
x\right)  =\chi_{p}\left(  \mathfrak{B}\left(  x,x^{\ast}\right)  \right)  $.
The Fourier transform \ is defined by%
\[
(\mathcal{F}g)(\xi)=%
{\displaystyle\int\limits_{\mathbb{Q}_{p}^{N}}}
g\left(  x\right)  \chi_{p}\left(  \mathfrak{B}\left(  x,\xi\right)  \right)
d\mu\left(  x\right)  \text{,}\quad\text{for }g\in L^{1},
\]
where $d\mu\left(  x\right)  $ is a Haar measure on $\mathbb{Q}_{p}^{N}$. Let
$\mathcal{L}\left(  \mathbb{Q}_{p}^{N}\right)  $ be the space of continuous
functions $g$ in $L^{1}$ whose Fourier transform $\mathcal{F}g$ is in $L^{1}$.
The measure $d\mu\left(  x\right)  $ can be normalized uniquely in such manner
that $(\mathcal{F}(\mathcal{F}g))(x)=g(-x)$ for every $g$ belonging to
$\mathcal{L}\left(  \mathbb{Q}_{p}^{N}\right)  $. We say that $d\mu\left(
x\right)  $ is \textit{\ a self-dual measure relative\ to} $\chi_{p}\left(
\left[  \cdot,\cdot\right]  \right)  $. Notice that $d\mu\left(  x\right)
=C(\mathfrak{q})d^{N}x$ where $C(\mathfrak{q})$ is a positive constant and
$d^{N}x$ is the Haar measure on $\mathbb{Q}_{p}^{N}$ normalized by the
condition $vol(B_{0}^{N})=1$. For further details about the material presented
in this section the reader may consult \cite{We1}.

We will also use the notation $\mathcal{F}_{x\rightarrow\xi}g$ and
$\widehat{g}$\ for the Fourier transform of $g$. The Fourier transform
$\mathcal{F}\left[  T\right]  $ of a distribution $T\in\mathcal{D}%
_{\mathbb{C}}^{^{\ast}}\left(  \mathbb{Q}_{p}^{N}\right)  $ is defined by%
\[
\left\langle g,\mathcal{F}\left[  T\right]  \right\rangle =\left\langle
\mathcal{F}\left[  g\right]  ,T\right\rangle \text{ for all }g\in
\mathcal{D}_{\mathbb{C}}(\mathbb{Q}_{p}^{N})\text{.}%
\]
The Fourier transform $f\rightarrow\mathcal{F}\left[  T\right]  $ is a linear
isomorphism from $\mathcal{D}_{\mathbb{C}}^{^{\ast}}\left(  \mathbb{Q}_{p}%
^{N}\right)  $\ onto itself. Furthermore, $T=\mathcal{F}\left[  \mathcal{F}%
\left[  T\right]  \left(  -\xi\right)  \right]  $.

\begin{remark}
\label{remark1}Given $r\in\left[  0,+\infty\right)  $, we denote by
$L^{r}\left(  \mathbb{Q}_{p}^{N},d^{N}x\right)  :=L^{r}$, the $\mathbb{C}%
$-vector space of all the complex valued functions $g$ satisfying
$\int_{\mathbb{Q}_{p}^{N}}\left\vert g\left(  x\right)  \right\vert ^{r}%
d^{N}x<\infty$. Let denote by $C^{_{\text{unif}}}:=C^{_{\text{unif}}}\left(
\mathbb{Q}_{p}^{N},\mathbb{C}\right)  $, the $\mathbb{C}$-vector space of all
the complex valued functions which are uniformly continuous. We denote by
$L_{\mathbb{R}}^{r}$, $C_{\mathbb{R}}^{_{\text{unif}}}$ the corresponding
$\mathbb{R}$-vector spaces.
\end{remark}

\section{\label{Sect3}A new class of non-Archimedean nuclear spaces}

The Bruhat-Schwartz space $\mathcal{D}_{\mathbb{C}}(\mathbb{Q}_{p}^{N})$ is
not invariant under the action of the pseudodifferential operators required in
this work. In this section we introduce a new type of nuclear countably
Hilbert spaces, which are invariant under the action of large class of
pseudodifferential operators.

\begin{remark}
We set $\mathbb{R}_{+}:=\left\{  x\in\mathbb{R}:x\geq0\right\}  $. We denote
by $\mathbb{N}$ the set of non-negative integers.
\end{remark}

\subsection{\label{Nuclear_spaces}A class of non-Archimedean nuclear spaces}

We define for $f$, $g$ in $\mathcal{D}_{\mathbb{R}}(\mathbb{Q}_{p}^{N})$ (or
in $\mathcal{D}_{\mathbb{C}}(\mathbb{Q}_{p}^{N})$) the following scalar
product:%
\begin{equation}
\left\langle f,g\right\rangle _{l,\alpha}:=\left\langle f,g\right\rangle _{l}=%
{\textstyle\int\limits_{\mathbb{Q}_{p}^{N}}}
\left[  \max\left(  1,\left\Vert \xi\right\Vert _{p}\right)  \right]
^{2\alpha l}\widehat{f}\left(  \xi\right)  \overline{\widehat{g}}\left(
\xi\right)  d^{N}\xi, \label{product}%
\end{equation}
for a fixed $\alpha\in\mathbb{R}_{+}\mathbb{\smallsetminus}\left\{  0\right\}
$ and $l\in\mathbb{Z}$. We also set $\left\Vert f\right\Vert _{l,\alpha}%
^{2}=:\left\Vert f\right\Vert _{l}^{2}=\left\langle f,f\right\rangle _{l}$.
Notice that $\left\Vert \cdot\right\Vert _{m}\leq\left\Vert \cdot\right\Vert
_{n}$ for $m\leq n$. Let denote by $\mathcal{H}_{\mathbb{R}}\left(
\mathbb{Q}_{p}^{N};l,\alpha\right)  \allowbreak=:\mathcal{H}_{\mathbb{R}%
}\left(  l\right)  $ the completion of $\mathcal{D}_{\mathbb{R}}%
(\mathbb{Q}_{p}^{N})$ with respect to $\left\langle \cdot,\cdot\right\rangle
_{l}$. Then $\mathcal{H}_{\mathbb{R}}\left(  n\right)  \subset\mathcal{H}%
_{\mathbb{R}}\left(  m\right)  $ for $m\leq n$. We set%
\[
\mathcal{H}_{\mathbb{R}}\left(  \mathbb{Q}_{p}^{N};\alpha,\infty\right)
:=\mathcal{H}_{\mathbb{R}}\left(  \mathbb{Q}_{p}^{N};\infty\right)
:=\mathcal{H}_{\mathbb{R}}\left(  \infty\right)  =%
{\textstyle\bigcap\limits_{l\in\mathbb{N}}}
\mathcal{H}_{\mathbb{R}}\left(  l\right)  .
\]
Notice that $\mathcal{H}_{\mathbb{R}}\left(  0\right)  =L_{\mathbb{R}}^{2}$
and that $\mathcal{H}_{\mathbb{R}}\left(  \infty\right)  \subset
L_{\mathbb{R}}^{2}$. With the topology induced by the family of seminorms
$\left\Vert \cdot\right\Vert _{l\in\mathbb{N}}$, $\mathcal{H}_{\mathbb{R}%
}\left(  \infty\right)  $ becomes a locally convex space, which is metrizable.
Indeed,
\[
d\left(  f,g\right)  :=\max_{l\in\mathbb{N}}\left\{  2^{-l}\frac{\left\Vert
f-g\right\Vert _{l}}{1+\left\Vert f-g\right\Vert _{l}}\right\}
\]
is a metric for the topology of $\mathcal{H}_{\mathbb{R}}\left(
\infty\right)  $ considered as a convex topological space, see e.g.
\cite{Gel-Vil}. A sequence $\left\{  f_{l}\right\}  _{l\in\mathbb{N}}$ in
$\left(  \mathcal{H}_{\mathbb{R}}\left(  \infty\right)  ,d\right)  $ converges
to $g\in\mathcal{H}_{\mathbb{R}}\left(  \infty\right)  $, if and only if,
$\left\{  f_{l}\right\}  _{l\in\mathbb{N}}$ converges to $g$ in the norm
$\left\Vert \cdot\right\Vert _{l}$ for all $l\in\mathbb{N}$. From this
observation follows that the topology on $\mathcal{H}_{\mathbb{R}}\left(
\infty\right)  $ coincides with the projective\ limit topology $\tau_{P}$. An
open neighborhood base at zero of $\tau_{P}$ is given by the choice of
$\epsilon>0$ and $l\in\mathbb{N}$, and the set
\[
U_{\epsilon,l}:=\left\{  f\in\mathcal{H}_{\mathbb{R}}\left(  \infty\right)
:\left\Vert f\right\Vert _{l}<\epsilon\right\}  .
\]

\begin{remark}
We denote by $\mathcal{H}_{\mathbb{C}}\left(  l\right)  $, $\mathcal{H}%
_{\mathbb{C}}\left(  \infty\right)  $ the $\mathbb{C}$-vector spaces
constructed from $\mathcal{D}_{\mathbb{C}}(\mathbb{Q}_{p}^{N})$. All the above
results are valid for these spaces. We shall use $d$ to denote the metric of
$\mathcal{H}_{\mathbb{C}}\left(  \infty\right)  $.
\end{remark}

\begin{lemma}
\label{lemma1} $\mathcal{H}_{\mathbb{R}}\left(  \infty\right)  $ endowed with
the topology $\tau_{P}$\ is a countably Hilbert space in the sense of Gel'fand
and Vilenkin, see e.g. \cite[Chapter I, Section 3.1]{Gel-Vil} or \cite[Section
1.2]{Obata}. Furthermore $\left(  \mathcal{H}_{\mathbb{R}}\left(
\infty\right)  ,\tau_{P}\right)  $ is metrizable and complete and hence a
Fréchet space.
\end{lemma}

\begin{proof}
By the previous considerations, it is sufficient to show that $\left\langle
\cdot,\cdot\right\rangle _{l\in\mathbb{N}}$ is a system of compatible scalar
products, i.e. if a sequence $\left\{  f_{l}\right\}  _{l\in\mathbb{N}}$ of
elements of $\mathcal{H}_{\mathbb{R}}\left(  \infty\right)  $ converges to
zero in the norm $\left\Vert \cdot\right\Vert _{m}$ and is a Cauchy sequence
in the norm $\left\Vert \cdot\right\Vert _{n}$, then it also converges to zero
in the norm $\left\Vert \cdot\right\Vert _{n}$. We may assume without loss of
generality that $m\leq n$ thus $\left\Vert \cdot\right\Vert _{m}\leq\left\Vert
\cdot\right\Vert _{n}$. By using $f_{l}$ $\underrightarrow{\left\Vert
\cdot\right\Vert _{m}}$ $0\in\mathcal{H}_{\mathbb{R}}\left(  m\right)  $ and
$f_{l}$ $\underrightarrow{\left\Vert \cdot\right\Vert _{n}}$ $f\in
\mathcal{H}_{\mathbb{R}}\left(  n\right)  \subset\mathcal{H}_{\mathbb{R}%
}\left(  m\right)  $, we conclude that $f=0$.
\end{proof}

\begin{lemma}
\label{lemma2}(i) Set $\overline{\left(  \mathcal{D}_{\mathbb{R}}%
(\mathbb{Q}_{p}^{N}),d\right)  }$ for the completion of the metric space
$\left(  \mathcal{D}_{\mathbb{R}}(\mathbb{Q}_{p}^{N}),d\right)  $. Then
$\overline{\left(  \mathcal{D}_{\mathbb{R}}(\mathbb{Q}_{p}^{N}),d\right)
}=\left(  \mathcal{H}_{\mathbb{R}}\left(  \infty\right)  ,d\right)  $.

\noindent(ii) $\left(  \mathcal{H}_{\mathbb{R}}\left(  \infty\right)
,d\right)  $ is a nuclear space.
\end{lemma}

\begin{proof}
Set $f\in\overline{\left(  \mathcal{D}_{\mathbb{R}}(\mathbb{Q}_{p}%
^{N}),d\right)  }$, then there exists a sequence $\left\{  f_{n}\right\}
_{n\in\mathbb{N}}$ in $\left(  \mathcal{D}_{\mathbb{R}}(\mathbb{Q}_{p}%
^{N}),d\right)  $ such that $f_{n}$ $\ \underrightarrow{\left\Vert
\cdot\right\Vert _{l}}$ $\ f$ for each $l\in\mathbb{N}$, i.e. $f\in\cap
_{l\in\mathbb{N}}\mathcal{H}_{l}$. Hence $\overline{\left(  \mathcal{D}%
_{\mathbb{R}}(\mathbb{Q}_{p}^{N}),d\right)  }\subset\left(  \mathcal{H}%
_{\mathbb{R}}\left(  \infty\right)  ,d\right)  $. Conversely, set
$g\in\mathcal{H}_{\mathbb{R}}\left(  \infty\right)  $. By using the density of
$\mathcal{D}_{\mathbb{R}}(\mathbb{Q}_{p}^{N})$ in $\mathcal{H}_{\mathbb{R}%
}\left(  l\right)  $, and the fact that $\left\Vert \cdot\right\Vert _{m}%
\leq\left\Vert \cdot\right\Vert _{n}$ if $m\leq n$ , we construct a sequence
$\left\{  g_{n}\right\}  _{n\in\mathbb{N}}$ in $\mathcal{D}_{\mathbb{R}%
}(\mathbb{Q}_{p}^{N})$ satisfying
\[
\left\Vert g_{n}-g\right\Vert _{l}\leq%
\begin{array}
[c]{cc}%
\frac{1}{n+1} & \text{for }0\leq l\leq n
\end{array}
\text{.}%
\]
Then $d\left(  g_{n},g\right)  \leq\max\left\{  \frac{1}{n+1},\frac{2}%
{n+1},\cdots,\frac{2^{-n}}{n+1},2^{-\left(  n+1\right)  }\right\}
\rightarrow0$ as $n\rightarrow\infty$. This fact shows that $g\in
\overline{\left(  \mathcal{D}_{\mathbb{R}}(\mathbb{Q}_{p}^{N}),d\right)  }$.

(ii) We recall that $\mathcal{D}_{\mathbb{C}}(\mathbb{Q}_{p}^{N})$ is a
nuclear space, cf. \cite[Section 4]{Bruhat}, and thus $\mathcal{D}%
_{\mathbb{R}}(\mathbb{Q}_{p}^{N})$ is a nuclear space, since any subspace of a
nuclear space is also nuclear, see e.g. \cite[Proposition 50.1]{Treves}. Now,
since the completion of a nuclear space is also nuclear, see e.g.
\cite[Proposition 50.1]{Treves}, by (i), $\mathcal{H}_{\mathbb{R}}\left(
\infty\right)  $ is a nuclear space.
\end{proof}

\begin{remark}
\label{nota2} (i) Lemma \ref{lemma2} is valid if we replace $\mathcal{D}%
_{\mathbb{R}}(\mathbb{Q}_{p}^{N})$ by $\mathcal{D}_{\mathbb{C}}(\mathbb{Q}%
_{p}^{N})$ and $\mathcal{H}_{\mathbb{R}}\left(  \infty\right)  $ by
$\mathcal{H}_{\mathbb{C}}\left(  \infty\right)  $.

\noindent(ii) By using the Cauchy-Schwartz inequality, if $l>\frac{N}{2\alpha
}$ and $g\in$ $\mathcal{H}_{\mathbb{R}}\left(  l\right)  $, then $\widehat
{g}\in L^{1}$, and thus $g\in C_{\mathbb{R}}^{_{\text{unif}}}$, see also
\cite[Lemma 16]{R-Zuniga}. Therefore $\mathcal{H}_{\mathbb{R}}\left(
\infty\right)  \subset L_{\mathbb{R}}^{2}\cap C_{\mathbb{R}}^{_{\text{unif}}}$.

\noindent(iii) In the definition of the product $\left\langle \cdot
,\cdot\right\rangle _{l}$, the weight $\left[  \max\left(  1,\left\Vert
\xi\right\Vert _{p}\right)  \right]  ^{2\alpha l}$ can be replaced by $\left[
1+\left\Vert \xi\right\Vert _{p}\right]  ^{2\alpha l}$ or by $\left[
1+\left\vert f\left(  \xi\right)  \right\vert _{p}\right]  ^{2\alpha l}$,
where $f$ is an elliptic polynomial, see Section \ref{Sect4} or \cite{Zu1},
and the corresponding norm is equivalent to $\left\Vert \cdot\right\Vert _{l}$.
\end{remark}

From Lemmas \ref{lemma1}-\ref{lemma2} we obtain the main result of this section.

\begin{theorem}
\label{Thm1} $\mathcal{H}_{\mathbb{R}}\left(  \infty\right)  $ is nuclear
countably Hilbert space.
\end{theorem}

\begin{remark}
\label{nota1}(i) As a nuclear Fréchet space $\mathcal{H}_{\mathbb{R}}\left(
\infty\right)  $ admits a sequence of defining Hilbertian norms $\left\vert
\cdot\right\vert _{m\in\mathbb{N}}$ such that (1) $\left\vert g\right\vert
_{m}\leq C_{m}\left\vert g\right\vert _{m+1}$, $g\in\mathcal{H}_{\mathbb{R}%
}\left(  \infty\right)  $, with some $C_{m}>0$; (2) the canonical map
$i_{n,n+1}:H_{\mathbb{R}}\left(  n+1\right)  \rightarrow H_{\mathbb{R}}\left(
n\right)  $ is of Hilbert-Schmidt type, where $H_{\mathbb{R}}\left(  n\right)
$ is the Hilbert space associated with $\left\vert \cdot\right\vert _{n}$, cf.
\cite[Proposition 1.3.2]{Obata}.

\noindent(ii) Let $\mathcal{H}_{\mathbb{R}}^{^{\ast}}\left(  l\right)  $ be
the dual space of $\mathcal{H}_{\mathbb{R}}\left(  l\right)  $. By identifying
$\mathcal{H}_{\mathbb{R}}^{^{\ast}}\left(  l\right)  $ with $\mathcal{H}%
_{\mathbb{R}}\left(  -l\right)  $ and denoting the dual pairing between
$\mathcal{H}_{\mathbb{R}}^{^{\ast}}\left(  \infty\right)  $ and $\mathcal{H}%
_{\mathbb{R}}\left(  \infty\right)  $ by $\left\langle \cdot,\cdot
\right\rangle $, we have from the results of Gel'fand and Vilenkin that
$\mathcal{H}_{\mathbb{R}}^{^{\ast}}\left(  \infty\right)  =\cup_{l\in
\mathbb{N}}\mathcal{H}_{\mathbb{R}}^{^{\ast}}\left(  l\right)  $. We shall
consider $\mathcal{H}_{\mathbb{R}}^{^{\ast}}\left(  \infty\right)  $ as
equipped with the weak topology.

\noindent(iii) Theorem \ref{Thm1} is valid for $\mathcal{H}_{\mathbb{C}%
}^{^{\ast}}\left(  \infty\right)  $.

\noindent(iv) There are some well-known results for constructing nuclear
countably Hilbert spaces starting with suitable operators, see e.g.
\cite[Sections 1.2-1.3]{Obata}. In order to use the mentioned results, and
taking into account the definition of our norms $\left\Vert \cdot\right\Vert
_{l}$, see (\ref{product}), we have to work with operators of type
\[
\boldsymbol{A}_{\gamma}\left(  \cdot\right)  =\mathcal{F}^{-1}(\left[
\max\left(  1,\left\Vert \xi\right\Vert _{p}\right)  \right]  ^{\gamma
}\mathcal{F}\left(  \cdot\right)  )\text{ \ with }\gamma\in\mathbb{R}.
\]
It is essential to show that $\boldsymbol{A}_{-\gamma}$ is a Hilbert-Schmidt
operator for some $\gamma>0$, see \cite[Proposition 1.3.4]{Obata}. By using
the orthonormal basis for $L^{2}\left(  \mathbb{Q}_{p}^{N}\right)  $
constructed by Albeverio and Kozyrev, see \cite[Theorem 1]{Alberio-Kozyrev},
one verifies that the elements of this basis are eigenfunctions of
$\boldsymbol{A}_{-\gamma}$ `with infinite multiplicity', for any $\gamma>0$.
This fact was also noted, in dimension one, in \cite[Theorem 2]{Su-Qiu}. Then
$\boldsymbol{A}_{-\gamma}$ is not a Hilbert-Schmidt operator for any
$\gamma>0$. The fact that $\mathcal{S}(\mathbb{R}^{n})$ is a nuclear space is
established by using the Hamiltonian of the harmonic oscillator and
\cite[Proposition 1.3.4]{Obata}, see for instance \cite[Appendix 5]{Hida et
al}.
\end{remark}

The following result will be used later on.

\begin{lemma}
\label{lemma2A} For any $l\in\mathbb{N}$, we set
\[
d\nu_{l,N}:=\left[  \max\left(  1,\left\Vert \xi\right\Vert _{p}\right)
\right]  ^{2\alpha l}d^{N}\xi,
\]
and
\[
L_{l,N}^{2}:=\left\{  f:\mathbb{Q}_{p}^{N}\text{ }\rightarrow\mathbb{C}:%
{\textstyle\int\nolimits_{\mathbb{Q}_{p}^{N}}}
\left\vert \widehat{f}\right\vert ^{2}d\nu_{l,N}<\infty\right\}  .
\]
Notice that $L_{l,N}^{2}\subset L^{2}$. Then $\mathcal{H}_{\mathbb{C}}\left(
l\right)  =L_{l,N}^{2}$ for any $l\in\mathbb{N}$. A similar result is
valid\ for $\mathcal{H}_{\mathbb{R}}\left(  l\right)  $.
\end{lemma}

\begin{proof}
Take $f\in\mathcal{H}_{\mathbb{C}}\left(  l\right)  $, then there exists a
sequence $\left\{  f_{n}\right\}  _{n\in\mathbb{N}}$ in $\mathcal{D}%
_{\mathbb{C}}\left(  \mathbb{Q}_{p}^{N}\right)  $ such that $f_{n}$
$\underrightarrow{\left\Vert \cdot\right\Vert _{l}}$ $f$, i.e. $\left[
\max\left(  1,\left\Vert \xi\right\Vert _{p}\right)  \right]  ^{\alpha
l}\widehat{f}_{n}$ $\underrightarrow{L^{2}}$ $\left[  \max\left(  1,\left\Vert
\xi\right\Vert _{p}\right)  \right]  ^{\alpha l}\widehat{f}$, and since
$L^{2}$ is complete, $\left[  \max\left(  1,\left\Vert \xi\right\Vert
_{p}\right)  \right]  ^{\alpha l}\widehat{f}\in L^{2}$, i.e. $f\in L_{l,N}%
^{2}$. Conversely, take $f\in L^{2}$ such that $\left[  \max\left(
1,\left\Vert \xi\right\Vert _{p}\right)  \right]  ^{\alpha l}\widehat{f}\in
L^{2}$. By using the fact that $\mathcal{D}_{\mathbb{C}}\left(  \mathbb{Q}%
_{p}^{N}\right)  $ is dense in $L^{2}$, there exists a sequence $\left\{
f_{n}\right\}  _{n\in\mathbb{N}}$ in $\mathcal{D}_{\mathbb{C}}\left(
\mathbb{Q}_{p}^{N}\right)  $ such that $f_{n}$ $\underrightarrow{L^{2}}$
$\left[  \max\left(  1,\left\Vert \xi\right\Vert _{p}\right)  \right]
^{\alpha l}\widehat{f}$. We now define $\widehat{g}_{n}=\frac{f_{n}}{\left[
\max\left(  1,\left\Vert \xi\right\Vert _{p}\right)  \right]  ^{\alpha l}}%
\in\mathcal{D}_{\mathbb{C}}\left(  \mathbb{Q}_{p}^{N}\right)  $. Then $g_{n}$
$\underrightarrow{\left\Vert \cdot\right\Vert _{l}}$ $f$, i.e. $f\in
\mathcal{H}_{\mathbb{C}}\left(  l\right)  $.
\end{proof}

\begin{remark}
\label{NoteL2}Since $L_{l,N}^{2}\subset L^{2}$, we have
\[
L_{l,N}^{2}=\left\{  \widehat{f}:\mathbb{Q}_{p}^{N}\text{ }\rightarrow
\mathbb{C}:%
{\textstyle\int\nolimits_{\mathbb{Q}_{p}^{N}}}
\left\vert \widehat{f}\right\vert ^{2}d\nu_{l,N}<\infty\right\}  ,
\]
which is isomorphic as Hilbert space to $L^{2}\left(  \mathbb{Q}_{p}^{N}%
,d\nu_{l,N}\right)  \subset L^{2}$ by means of the Fourier transform. A
similar assertion is valid in the real case.
\end{remark}

\subsubsection{Some remarks on tensor products}

Let $\mathfrak{X}$ and $\mathfrak{N}$ be locally convex spaces. We denote by
$\mathfrak{X}\otimes_{\text{alg}}\mathfrak{N}$ their algebraic tensor product.
The $\pi$\textit{-topology} on $\mathfrak{X}\otimes_{\text{alg}}\mathfrak{N}%
$\ is the strongest locally convex topology such that the canonical map
$\mathfrak{X}\times\mathfrak{N}\rightarrow\mathfrak{X}\otimes_{\text{alg}%
}\mathfrak{N}$ is continuous. The completion of $\mathfrak{X}\otimes
_{\text{alg}}\mathfrak{N}$\ with respect to the $\pi$-topology is called the
$\pi$\textit{-tensor product} and is denoted by $\mathfrak{X}\otimes_{\pi
}\mathfrak{N}$. If $\mathfrak{X}$ and $\mathfrak{N}$ are nuclear spaces, then
$\mathfrak{X}\otimes_{\pi}\mathfrak{N}$ is also nuclear, cf. \cite[Proposition
1.3.7]{Obata}.

Let $\mathcal{L}$ and $\mathcal{K}$ be Hilbert spaces, we denote by
$\mathcal{L}\otimes\mathcal{K}$ the Hilbert space tensor product of
$\mathcal{L}$ and $\mathcal{K}$. Now, let $\mathfrak{X}$ and $\mathfrak{N}$ be
locally convex spaces with defining Hilbertian seminorms $\left\{  \left\Vert
\cdot\right\Vert _{\alpha\in A}\right\}  $ and $\left\{  \left\Vert
\cdot\right\Vert _{\beta\in B}^{\prime}\right\}  $ respectively. Let
$\mathfrak{X}_{\alpha}$ and $\mathfrak{N}_{\beta}$ be the Hilbert spaces
associated with $\left\Vert \cdot\right\Vert _{\alpha}$\ and $\left\Vert
\cdot\right\Vert _{\beta}^{\prime}$\ respectively. Then $\left\{
\mathfrak{X}_{\alpha}\otimes\mathfrak{N}_{\beta}\right\}  _{\alpha\in
A,\beta\in B}$ becomes a projective system of Hilbert spaces. If
$\mathfrak{X}$ or $\mathfrak{N}$ is nuclear, then
\begin{equation}
\mathfrak{X}\otimes_{\pi}\mathfrak{N}\cong\projlim_{\alpha,\beta}\left(
\mathfrak{X}_{\alpha}\otimes\mathfrak{N}_{\beta}\right)  , \label{projlim}%
\end{equation}
cf. \cite[Proposition 1.3.8]{Obata}.

\begin{lemma}
\label{lemma_tensor_product}$\mathcal{H}_{\mathbb{R}}\left(  \mathbb{Q}%
_{p}^{N};\infty\right)  \otimes_{\pi}\mathcal{H}_{\mathbb{R}}\left(
\mathbb{Q}_{p}^{M};\infty\right)  \cong\mathcal{H}_{\mathbb{R}}\left(
\mathbb{Q}_{p}^{N+M};\infty\right)  $ for $N$, $M\in\mathbb{N\smallsetminus
}\left\{  0\right\}  $. By induction $\mathcal{H}_{\mathbb{R}}\left(
\mathbb{Q}_{p}^{N_{1}};\infty\right)  \otimes_{\pi}\cdots\otimes_{\pi
}\mathcal{H}_{\mathbb{R}}\left(  \mathbb{Q}_{p}^{N_{k}};\infty\right)
\cong\mathcal{H}_{\mathbb{R}}\left(  \mathbb{Q}_{p}^{\sum_{i=1}^{k}N_{i}%
};\infty\right)  $ for $N_{i}\in\mathbb{N\smallsetminus}\left\{  0\right\}  $,
$i=1,\cdots,k$.
\end{lemma}

\begin{proof}
By using (\ref{projlim}),%
\begin{equation}
\mathcal{H}_{\mathbb{R}}\left(  \mathbb{Q}_{p}^{N};\infty\right)  \otimes
_{\pi}\mathcal{H}_{\mathbb{R}}\left(  \mathbb{Q}_{p}^{M};\infty\right)
\cong\projlim_{l,m}\left(  \mathcal{H}_{\mathbb{R}}\left(  \mathbb{Q}_{p}%
^{N};l\right)  \otimes\mathcal{H}_{\mathbb{R}}\left(  \mathbb{Q}_{p}%
^{M};m\right)  \right)  . \label{key_morphism}%
\end{equation}
By Remark \ref{NoteL2}%
\[
\mathcal{H}_{\mathbb{R}}\left(  \mathbb{Q}_{p}^{L};l\right)  \cong
L_{\mathbb{R}}^{2}\left(  \mathbb{Q}_{p}^{L},\left[  \max\left(  1,\left\Vert
\xi\right\Vert _{p}\right)  \right]  ^{2\alpha l}d^{L}\xi\right)  ,
\]
then%
\begin{multline*}
\mathcal{H}_{\mathbb{R}}\left(  \mathbb{Q}_{p}^{N};l\right)  \otimes
\mathcal{H}_{\mathbb{R}}\left(  \mathbb{Q}_{p}^{M};m\right)  \cong\\
L_{\mathbb{R}}^{2}\left(  \mathbb{Q}_{p}^{N+M},\left[  \max\left(
1,\left\Vert \xi\right\Vert _{p}\right)  \right]  ^{2\alpha l}\left[
\max\left(  1,\left\Vert \zeta\right\Vert _{p}\right)  \right]  ^{2\alpha
m}d^{N}\xi d^{M}\zeta\right) \\
=:L_{0}^{2}\left(  N+M;l,m\right)  .
\end{multline*}

Set $\left(  l,m\right)  \geq\left(  l^{\prime},m^{\prime}\right)
\Leftrightarrow l\geq l^{\prime}$ and $m\geq m^{\prime}$, and define%
\[
j_{\left(  l,m\right)  ,\left(  l^{\prime},m^{\prime}\right)  }:L_{0}%
^{2}\left(  N+M;l,m\right)  \hookrightarrow L_{0}^{2}\left(  N+M;l^{\prime
},m^{\prime}\right)  \text{ for }\left(  l,m\right)  \geq\left(  l^{\prime
},m^{\prime}\right)  \text{, }%
\]
where `$\hookrightarrow$' denotes a continuous embedding of Hilbert spaces,
then
\[
\left(  \left\{  L_{0}^{2}\left(  N+M;l,m\right)  \right\}  _{\left(
l,m\right)  },\left\{  j_{\left(  l,m\right)  ,\left(  l^{\prime},m^{\prime
}\right)  }\right\}  _{\left(  l,m\right)  \geq\left(  l^{\prime},m^{\prime
}\right)  }\right)
\]
forms a projective system and by (\ref{key_morphism}),
\[
\mathcal{H}_{\mathbb{R}}\left(  \mathbb{Q}_{p}^{N};\infty\right)  \otimes
_{\pi}\mathcal{H}_{\mathbb{R}}\left(  \mathbb{Q}_{p}^{M};\infty\right)
\cong\projlim_{l,m}L_{0}^{2}\left(  N+M;l,m\right)  .
\]
We now set
\[
L_{1}^{2}\left(  N+M,k\right)  :=L_{\mathbb{R}}^{2}\left(  \mathbb{Q}%
_{p}^{N+M},\left[  \max\left(  1,\left\Vert \xi\right\Vert _{p},\left\Vert
\zeta\right\Vert _{p}\right)  \right]  ^{2\alpha k}d^{N}\xi d^{M}\zeta\right)
,
\]
where $\xi=\left(  \xi_{1},\cdots,\xi_{N}\right)  \in\mathbb{Q}_{p}^{N}$,
$\zeta=\left(  \zeta_{1},\cdots,\zeta_{M}\right)  \in\mathbb{Q}_{p}^{M}$. We
now note that%
\begin{equation}
L_{1}^{2}\left(  N+M,l+m\right)  \hookrightarrow L_{0}^{2}\left(
N+M;l,m\right)  \hookrightarrow L_{1}^{2}\left(  N+M,\min\left\{  l,m\right\}
\right)  , \label{Key_emmbeding}%
\end{equation}
for $l,m\in\mathbb{N}$.

On the other hand, since \ for $l\geq k$, the inequality $\left\Vert
\cdot\right\Vert _{k}\leq\left\Vert \cdot\right\Vert _{l}$ implies
\[
i_{k,l}:L_{1}^{2}\left(  N+M,l\right)  \hookrightarrow L_{1}^{2}\left(
N+M,k\right)  \text{ for }k\leq l,
\]
thus $\left(  \left\{  L_{1}^{2}\left(  N+M,k\right)  \right\}  _{k},\left\{
i_{k,l}\right\}  _{k\leq l}\right)  $ forms a projective system. The result
follows by showing that%
\[
\projlim_{\left(  l,m\right)  }L_{0}^{2}\left(  N+M;l,m\right)  \cong
\projlim_{l+m}L_{1}^{2}\left(  N+M;l+m\right)
\]
from (\ref{Key_emmbeding}).
\end{proof}

\begin{remark}
\label{Kernel_theorem}Let $\mathfrak{X}$ be a nuclear Fréchet space and let
$\mathfrak{N}$ be a Fréchet space. We denote by $\mathcal{B}(\mathfrak{X}%
,\mathfrak{N})$ the space of jointly continuous bilinear maps from
$\mathfrak{X}\times\mathfrak{N}$ into $\mathbb{R}$, and by $\mathcal{B}%
_{\text{sep}}(\mathfrak{X},\mathfrak{N})$ the space of separately continuous
bilinear maps from $\mathfrak{X}\times\mathfrak{N}$ into $\mathbb{R}$. The
following version of the Kernel Theorem will be used later on:%
\[
\left(  \mathfrak{X}\otimes_{\pi}\mathfrak{N}\right)  ^{\ast}\cong
\mathcal{B}(\mathfrak{X},\mathfrak{N})\cong\mathcal{B}_{\text{sep}%
}(\mathfrak{X},\mathfrak{N}),
\]
cf. \cite[Theorem 1.3.10 and Proposition 1.3.12.]{Obata}
\end{remark}

\subsection{Pseudodifferential operators acting on $\mathcal{H}_{\mathbb{C}%
}\left(  \infty\right)  $}

\begin{definition}
\label{def1}We say that a function $\mathfrak{a}:\mathbb{Q}_{p}^{N}%
\rightarrow\mathbb{R}_{+}$ is a smooth symbol, if it satisfies the following properties:

\noindent(i) $\mathfrak{a}$ is a continuous function;

\noindent(ii) there exists a positive constant $C=C(\mathfrak{a})$ such that
$\mathfrak{a}\left(  \xi\right)  \geq C$ for any $\xi\in\mathbb{Q}_{p}^{N}$;

\noindent(iii) there exist positive constants $C_{0}$, $C_{1}$, $\alpha$,
$m_{0}$, with $m_{0}\in\mathbb{N}$, such that
\[
C_{0}\left\Vert \xi\right\Vert _{p}^{\alpha}\leq\mathfrak{a}\left(
\xi\right)  \leq C_{1}\left\Vert \xi\right\Vert _{p}^{\alpha}\text{ \ for
}\left\Vert \xi\right\Vert _{p}\geq p^{m_{0}}\text{.}%
\]

\end{definition}

Given a smooth symbol $\mathfrak{a}\left(  \xi\right)  $, we attach to it the
following pseudodifferential operator:%
\[%
\begin{array}
[c]{ccc}%
\mathcal{D}_{\mathbb{C}}(\mathbb{Q}_{p}^{N}) & \rightarrow & L^{2}\cap
C^{\text{unif}}\\
g & \rightarrow & \boldsymbol{A}g,
\end{array}
\]
where $\left(  \boldsymbol{A}g\right)  \left(  x\right)  =\mathcal{F}%
_{\xi\rightarrow x}^{-1}\left(  \mathfrak{a}\left(  \xi\right)  \mathcal{F}%
_{x\rightarrow\xi}g\right)  $.

\begin{lemma}
\label{lemma3} For any $l\in\mathbb{N}$, the mapping $\boldsymbol{A}%
:\mathcal{H}_{\mathbb{C}}\left(  l+1\right)  \rightarrow\mathcal{H}%
_{\mathbb{C}}\left(  l\right)  $ is a well-defined continuous mapping between
Banach spaces.
\end{lemma}

\begin{proof}
Let $g\in\mathcal{D}_{\mathbb{C}}(\mathbb{Q}_{p}^{N})$, then%
\begin{multline*}
\left\Vert \boldsymbol{A}g\right\Vert _{l}^{2}\leq%
{\textstyle\int\limits_{B_{m_{0}}^{N}}}
\left[  \max\left(  1,\left\Vert \xi\right\Vert _{p}\right)  \right]
^{2\alpha l}\left\vert \mathfrak{a}\left(  \xi\right)  \right\vert
^{2}\left\vert \widehat{g}\left(  \xi\right)  \right\vert ^{2}d^{N}\xi\\
+C_{1}^{2}%
{\textstyle\int\limits_{\mathbb{Q}_{p}^{N}\smallsetminus B_{m_{0}}^{N}}}
\left\Vert \xi\right\Vert _{p}^{2\alpha\left(  l+1\right)  }\left\vert
\widehat{g}\left(  \xi\right)  \right\vert ^{2}d^{N}\xi\leq\left(  \sup
_{\xi\in B_{m_{0}}^{N}}\left[  \mathfrak{a}\left(  \xi\right)  \max\left(
1,\left\Vert \xi\right\Vert _{p}\right)  \right]  ^{\alpha l}\right)
^{2}\left\Vert g\right\Vert _{0}^{2}\\
+C_{1}^{2}\left\Vert g\right\Vert _{l+1}^{2}\leq C_{2}\left\Vert g\right\Vert
_{l+1}^{2}.
\end{multline*}
Now, by Lemma \ref{lemma2A}, $\boldsymbol{A}g\in\mathcal{H}_{\mathbb{C}%
}\left(  l\right)  $. The result follows from the density of $\mathcal{D}%
_{\mathbb{C}}(\mathbb{Q}_{p}^{N})$ in $\mathcal{H}_{\mathbb{C}}\left(
l+1\right)  $.
\end{proof}

\begin{lemma}
\label{lemma5} For any $l\in\mathbb{N}$, the mapping $\boldsymbol{A}%
:\mathcal{H}_{\mathbb{C}}\left(  l+1\right)  \rightarrow\mathcal{H}%
_{\mathbb{C}}\left(  l\right)  $ has a continuous inverse defined on
$\mathcal{H}_{\mathbb{C}}\left(  l\right)  $. In particular, $\boldsymbol{A}$
is a bi-continuous bijection between Banach spaces.
\end{lemma}

\begin{proof}
Take $g\in\mathcal{D}_{\mathbb{C}}(\mathbb{Q}_{p}^{N})\subset\mathcal{H}%
_{\mathbb{C}}\left(  l\right)  $ and set $u:=\mathcal{F}^{-1}\left(
\frac{\widehat{g}\left(  \xi\right)  }{a\left(  \xi\right)  }\right)  $. Then
$u\in L^{2}\cap C^{\text{unif}}$, and it is the unique solution of
$\boldsymbol{A}u=g$. Now
\begin{align*}
\left\Vert u\right\Vert _{l+1}^{2}  &  =%
{\textstyle\int\limits_{B_{m_{0}}^{N}}}
\left[  \max\left(  1,\left\Vert \xi\right\Vert _{p}\right)  \right]
^{2\alpha\left(  l+1\right)  }\frac{\left\vert \widehat{g}\left(  \xi\right)
\right\vert ^{2}}{\left\vert \mathfrak{a}\left(  \xi\right)  \right\vert ^{2}%
}d^{N}\xi\\
+%
{\textstyle\int\limits_{\mathbb{Q}_{p}^{N}\smallsetminus B_{m_{0}}^{N}}}
\left\Vert \xi\right\Vert _{p}^{2\alpha\left(  l+1\right)  }\frac{\left\vert
\widehat{g}\left(  \xi\right)  \right\vert ^{2}}{\left\vert \mathfrak{a}%
\left(  \xi\right)  \right\vert ^{2}}d^{N}\xi &  \leq C^{\prime}\left\Vert
g\right\Vert _{0}^{2}+\frac{1}{C_{0}^{2}}%
{\textstyle\int\limits_{\mathbb{Q}_{p}^{N}\smallsetminus B_{m_{0}}^{N}}}
\left\Vert \xi\right\Vert _{p}^{2\alpha l}\left\vert \widehat{g}\left(
\xi\right)  \right\vert ^{2}d^{N}\xi\\
&  \leq C^{\prime}\left\Vert g\right\Vert _{0}^{2}+\frac{1}{C_{0}^{2}%
}\left\Vert g\right\Vert _{l}^{2}\leq C^{\prime\prime}\left\Vert g\right\Vert
_{l}^{2}.
\end{align*}
By Lemma \ref{lemma2A}, $u\in\mathcal{H}_{\mathbb{C}}\left(  l+1\right)  $ and
since $\mathcal{D}_{\mathbb{C}}(\mathbb{Q}_{p}^{N})$ is dense in
$\mathcal{H}_{\mathbb{C}}\left(  l\right)  $, the mapping
\[%
\begin{array}
[c]{cccc}%
\boldsymbol{A}^{-1} & \mathcal{H}_{\mathbb{C}}\left(  l\right)  & \rightarrow
& \mathcal{H}_{\mathbb{C}}\left(  l+1\right) \\
& g & \rightarrow & u,
\end{array}
\]
is well-defined and continuous.
\end{proof}

\begin{theorem}
\label{Thm2} (i) The mapping $\boldsymbol{A}:\mathcal{H}_{\mathbb{C}}\left(
\infty\right)  \rightarrow\mathcal{H}_{\mathbb{C}}\left(  \infty\right)  $ is
a bi-continuous isomorphism of locally convex spaces. (ii) $\mathcal{H}%
_{\mathbb{C}}\left(  \infty\right)  \subset L^{\infty}\cap C^{\text{unif}}\cap
L^{1}\cap L^{2}$.
\end{theorem}

\begin{proof}
(i) By by Lemma \ref{lemma3}, $\boldsymbol{A}$ is a well-defined mapping. In
addition, by Lemma \ref{lemma5}, $\boldsymbol{A}$ is a bijection from
$\mathcal{H}_{\mathbb{C}}\left(  \infty\right)  $ onto itself. To verify the
continuity of $\boldsymbol{A}$, we take a sequence $\left\{  g_{n}\right\}
_{n\in\mathbb{N}}$ in $\mathcal{H}_{\mathbb{C}}\left(  \infty\right)  $ such
that $g_{n}$ $\ \underrightarrow{d}$ $\ g$, with $g\in\mathcal{H}_{\mathbb{C}%
}\left(  \infty\right)  $, i.e. $g_{n}$ $\ \underrightarrow{\left\Vert
\cdot\right\Vert _{m}}$ $\ g$, for all $m\in\mathbb{N}$. Take $m=l+1$,
$g_{n}\in\mathcal{H}_{\mathbb{C}}\left(  l+1\right)  $, then by Lemma
\ref{lemma5}, $\boldsymbol{A}g_{n}$ $\ \underrightarrow{\left\Vert
\cdot\right\Vert _{l+1}}$ $\ \boldsymbol{A}g\in\mathcal{H}_{\mathbb{C}}\left(
l\right)  $, for any $l\in\mathbb{N}$. The case $\boldsymbol{A}g_{n}$
$\ \underrightarrow{\left\Vert \cdot\right\Vert _{0}}$ $\ \boldsymbol{A}g$ is
verified directly. Therefore $\boldsymbol{A}g_{n}$ $\ \underrightarrow{d}$
$\ \boldsymbol{A}g$. We now show that $\boldsymbol{A}^{-1}$ is continuous.
Take a sequence $\left\{  g_{n}\right\}  _{n\in\mathbb{N}}$ in $\mathcal{H}%
_{\mathbb{C}}\left(  \infty\right)  $ as before. By Lemma \ref{lemma5}, there
exists a unique sequence $\left\{  u_{n}\right\}  _{n\in\mathbb{N}}$ such that
$\boldsymbol{A}u_{n}=g_{n}$, with $g_{n}\in\mathcal{H}_{\mathbb{C}}\left(
l\right)  $ and $u_{n}\in\mathcal{H}_{\mathbb{C}}\left(  l+1\right)  $, for
any $l\in\mathbb{N}$. By verifying that $\left\{  u_{n}\right\}
_{n\in\mathbb{N}}$ is Cauchy in $\left\Vert \cdot\right\Vert _{m}$ for any
$m\in\mathbb{N}$, there exists $u\in\mathcal{H}_{\mathbb{C}}\left(
\infty\right)  $ such that $u_{n}$ $\ \underrightarrow{d}$ $\ u$, i.e.
$\boldsymbol{A}^{-1}g_{n}$ $\ \underrightarrow{d}$ $\ u$. By the continuity of
$\boldsymbol{A}$, we have $\boldsymbol{A}u=g$, \ now by using that
$\boldsymbol{A}$ is a bijection on $\mathcal{H}_{\mathbb{C}}\left(
\infty\right)  $, we conclude that $\boldsymbol{A}^{-1}g_{n}$
$\ \underrightarrow{d}$ $\ \boldsymbol{A}^{-1}g$.

(ii) In Remark \ref{nota2} (ii), we already noted that for any function $f$ in
$\mathcal{H}_{\mathbb{C}}\left(  \infty\right)  $, $\widehat{f}$ is
integrable, and thus $f\in C^{\text{unif}}$, and by the Riemann-Lebesgue
Theorem $f\in L^{\infty}$. We now define for $\alpha>0$, the operator
\[
\left(  \widetilde{D}^{\alpha}f\right)  \left(  x\right)  =\mathcal{F}%
_{\xi\rightarrow x}^{-1}\left(  \left[  \max\left(  1,\left\Vert
\xi\right\Vert _{p}\right)  \right]  ^{\alpha}\mathcal{F}_{x\rightarrow\xi
}\left(  f\right)  \right)  .
\]
By part (i), $\widetilde{D}^{\alpha}$ gives rise \ to a bicontinuous
isomorphism of $\mathcal{H}_{\mathbb{C}}\left(  \infty\right)  $. Then, for
any $f\in\mathcal{H}_{\mathbb{C}}\left(  \infty\right)  $, $\widetilde
{D}^{\alpha}f\in\mathcal{H}_{\mathbb{C}}\left(  \infty\right)  $, and thus
$\widehat{\widetilde{D}^{\alpha}f}=\left[  \max\left(  1,\left\Vert
\xi\right\Vert _{p}\right)  \right]  ^{\alpha}\widehat{f}\in L^{1}$, and by
Riemann-Lebesgue Theorem there are positive constants $C_{0}$, $C_{1}$ such
that
\begin{equation}
\left\vert \widehat{f}\left(  \xi\right)  \right\vert \leq\frac{C_{0}}{\left[
\max\left(  1,C_{1}\left\Vert \xi\right\Vert _{p}\right)  \right]  ^{\alpha}}.
\label{ineqq_f}%
\end{equation}

For $k\in\mathbb{N}$, we set%
\[
A_{k}\left(  f\chi_{p}\left(  \mathcal{B}\left(  \cdot,\xi\right)  \right)
\right)  =%
{\textstyle\int\nolimits_{\mathbb{Q}_{p}^{N}}}
\Omega\left(  p^{-k}\left\Vert x\right\Vert _{p}\right)  f\left(  x\right)
\chi_{p}\left(  \mathcal{B}\left(  x,\xi\right)  \right)  d\mu\left(
x\right)  ,
\]
where $\Omega\left(  p^{-k}\left\Vert x\right\Vert _{p}\right)  $ is the
characteristic function of the ball $B_{k}^{N}$. By using that $\widehat
{\left[  \widehat{f}\left(  -y\right)  \right]  }=f\left(  y\right)  $,
$d\mu\left(  x\right)  =C\left(  \mathfrak{q}\right)  d^{N}x$ and that
\[
\mathcal{F}_{x\rightarrow y}\left(  \Omega\left(  p^{-k}\left\Vert
x\right\Vert _{p}\right)  \chi_{p}\left(  \mathcal{B}\left(  x,\xi\right)
\right)  \right)  =C\left(  \mathfrak{q}\right)  p^{kN}\Omega\left(
p^{k+\beta}\left\Vert y+\xi\right\Vert _{p}\right)  ,
\]
where $\beta$\ is constant that depends only on $\mathcal{B}$, see
(\ref{Formula_B}), we have%
\[
A_{k}\left(  f\chi_{p}\left(  \mathcal{B}\left(  \cdot,\xi\right)  \right)
\right)  =C\left(  \mathfrak{q}\right)  p^{kN}%
{\textstyle\int\limits_{\left\Vert \xi-y\right\Vert _{p}\leq p^{-k-\beta}}}
\widehat{f}\left(  y\right)  d^{N}y.
\]
Now from (\ref{ineqq_f}),
\[
p^{kN}%
{\textstyle\int\limits_{\left\Vert \xi-y\right\Vert _{p}\leq p^{-k-\beta}}}
\left\vert \widehat{f}\left(  y\right)  \right\vert d^{N}y\leq\frac
{C_{0}C\left(  \mathfrak{q}\right)  p^{-N\beta}}{\left[  \max\left(
1,C_{1}\left\Vert \xi\right\Vert _{p}\right)  \right]  ^{\alpha}}%
\]
for $k$ big enough. Therefore $\lim_{k\rightarrow\infty}A_{k}\left(  f\chi
_{p}\left(  \mathcal{B}\left(  \cdot,\xi\right)  \right)  \right)  $ exists
for every $\xi\in\mathbb{Q}_{p}^{N}$, in particular $\lim_{k\rightarrow\infty}%
{\textstyle\int\nolimits_{\left\Vert x\right\Vert _{p}\leq p^{k}}}
f\left(  x\right)  d^{N}x$ exists, which means that $f\in L^{1}$.
\end{proof}

\section{\label{Sect4}Pseudodifferential Operators and Green Functions}

We take $\mathfrak{l}\left(  \xi\right)  \in\mathbb{Z}_{p}\left[  \xi
_{1},\cdots,\xi_{N}\right]  $ to be an\textit{ elliptic polynomial} of degree
$d$, this means that $\mathfrak{l}$ is homogeneous of degree $d$ and satisfies
$\mathfrak{l}\left(  \xi\right)  =0\Leftrightarrow\xi=0$. There are
infi\-nitely many elliptic polynomials. We consider the following
\textit{elliptic pseudodifferential operator}:%
\[
\left(  \mathbf{L}_{\alpha}h\right)  \left(  x\right)  =\mathcal{F}%
_{\xi\rightarrow x}^{-1}\left(  \left\vert \mathfrak{l}\left(  \xi\right)
\right\vert _{p}^{\alpha}\mathcal{F}_{x\rightarrow\xi}h\right)  ,
\]
where $\alpha>0$ and $h\in\mathcal{D}_{\mathbb{C}}(\mathbb{Q}_{p}^{N})$.

We shall call a fundamental solution $G\left(  x;m,\alpha\right)  $ of the
equation%
\begin{equation}
\left(  \boldsymbol{L}_{\alpha}+m^{2}\right)  u=h,\;\text{with }%
h\in\mathcal{D}_{\mathbb{C}}(\mathbb{Q}_{p}^{N}),\text{ }m>0,
\label{equation1}%
\end{equation}
a \textit{Green function} of $\boldsymbol{L}_{\alpha}$. As a distribution \ on
$\mathcal{D}_{\mathbb{C}}(\mathbb{Q}_{p}^{N})$, the Green function is given
by
\begin{equation}
G\left(  x;m,\alpha\right)  =\mathcal{F}_{\xi\rightarrow x}^{-1}\left(
\frac{1}{\left\vert \mathfrak{l}\left(  \xi\right)  \right\vert _{p}^{\alpha
}+m^{2}}\right)  . \label{GreenFunc}%
\end{equation}
Notice that since
\begin{equation}
C_{0}^{\alpha}\left\Vert \xi\right\Vert _{p}^{\alpha d}\leq\left\vert
\mathfrak{l}\left(  \xi\right)  \right\vert ^{\alpha}\leq C_{1}^{\alpha
}\left\Vert \xi\right\Vert _{p}^{\alpha d}, \label{basic_estimate}%
\end{equation}
for some positive constants $C_{0}$, $C_{1}$, cf. \cite[Lemma 1]{Zu1},
\[
\frac{1}{\left\vert \mathfrak{l}\left(  \xi\right)  \right\vert _{p}^{\alpha
}+m^{2}}\in L^{1}\left(  \mathbb{Q}_{p}^{N},d^{N}\xi\right)  \text{\ \ for
\ }\alpha d>N,
\]
and in this case, $G\left(  x;m,\alpha\right)  $ is an $L^{\infty}$-function.

\subsection{Green functions on $\mathcal{D}_{\mathbb{C}}^{^{\ast}}%
(\mathbb{Q}_{p}^{N})$}

\begin{proposition}
\label{Prop1} The Green function $G\left(  x;m,\alpha\right)  $ verifies the
following properties:

\noindent(i) the function $G\left(  x;m,\alpha\right)  $ is continuous on
$\mathbb{Q}_{p}^{N}\smallsetminus\left\{  0\right\}  $;

\noindent(ii) if $\alpha d>N$, then the function $G\left(  x;m,\alpha\right)
$ is continuous;

\noindent(iii) for $0<\alpha d\leq N$, the function $G\left(  x;m,\alpha
\right)  $ is locally constant on $\mathbb{Q}_{p}^{N}\smallsetminus\left\{
0\right\}  $, and
\[
\left\vert G\left(  x;m,\alpha\right)  \right\vert \leq\left\{
\begin{array}
[c]{lll}%
C\left\Vert x\right\Vert _{p}^{2\alpha d-N} & \text{for} & 0<\alpha d<N\\
&  & \\
C_{0}-C_{1}\ln\left\Vert x\right\Vert _{p} & \text{for} & N=\alpha d,
\end{array}
\right.
\]
for $\left\Vert x\right\Vert _{p}\leq1$, where $C$, $C_{0}$, $C_{1}$ are
positive constants,

\noindent(iv) $\left\vert G\left(  x;m,\alpha\right)  \right\vert \leq
C_{2}\left\Vert x\right\Vert _{p}^{-\alpha d-N}$ as $\left\Vert x\right\Vert
_{p}\rightarrow\infty$, where $C_{2}$ is positive constant;

\noindent(v) $G\left(  x;m,\alpha\right)  \geq0$ on $\mathbb{Q}_{p}%
^{N}\smallsetminus\left\{  0\right\}  $.
\end{proposition}

\begin{proof}
(i) We first notice that
\begin{equation}
G\left(  x;m,\alpha\right)  =%
{\displaystyle\sum\limits_{l=-\infty}^{+\infty}}
{\LARGE g}^{\left(  l\right)  }\left(  x;m,\alpha\right)  , \label{GreenFunc2}%
\end{equation}
where%
\[
{\LARGE g}^{\left(  l\right)  }\left(  x;m,\alpha\right)  =%
{\displaystyle\int\limits_{p^{-l}S_{0}^{N}}}
\frac{\chi_{p}\left(  -\mathcal{B}\left(  x,\xi\right)  \right)  }{\left\vert
\mathfrak{l}\left(  \xi\right)  \right\vert _{p}^{\alpha}+m^{2}}d^{N}\xi,
\]
here $S_{0}^{N}=\left\{  \xi\in\mathbb{Z}_{p}^{N}:\left\Vert \xi\right\Vert
_{p}=1\right\}  $. For $x\in\mathbb{Q}_{p}^{N}\smallsetminus\left\{
0\right\}  $, we take $x=p^{ord\left(  x\right)  }x_{0}$, $x_{0}=\left(
x_{0,1},\cdots,x_{0,N}\right)  $ with $\left\Vert x_{0}\right\Vert _{p}=1$. By
using the standard basis of $\mathbb{Q}_{p}^{N}$ we have%
\[
\mathcal{B}\left(  x,\xi\right)  =p^{ord\left(  x\right)  }\mathcal{B}\left(
x_{0},\xi\right)  =p^{ord\left(  x\right)  }%
{\displaystyle\sum\limits_{i=1}^{N}}
\xi_{i}\left[
{\displaystyle\sum\limits_{j=1}^{N}}
B_{ij}x_{0,j}\right]
\]
where the $B_{ij}\in\mathbb{Q}_{p}$, and $\det\left[  B_{ij}\right]  \neq0$.
By considering $\left[  B_{ij}\right]  $ as a vector in $\mathbb{Q}_{p}%
^{N^{2}}$ we have $\left[  B_{ij}\right]  =p^{\beta}\left[  \widetilde{B}%
_{ij}\right]  $, where the vector $\left[  \widetilde{B}_{ij}\right]
\in\mathbb{Q}_{p}^{N^{2}}$ has norm $1$. Then
\begin{equation}
\mathcal{B}\left(  x,\xi\right)  =p^{ord\left(  x\right)  +\beta}%
{\displaystyle\sum\limits_{i=1}^{N}}
\xi_{i}\left[
{\displaystyle\sum\limits_{j=1}^{N}}
\widetilde{B}_{ij}x_{0,j}\right]  =p^{ord\left(  x\right)  +\beta}%
{\displaystyle\sum\limits_{i=1}^{N}}
\xi_{i}\widetilde{A}_{i}\text{, } \label{Formula_B}%
\end{equation}

with $\left\Vert \left(  \widetilde{A}_{1},\cdots,\widetilde{A}_{N}\right)
\right\Vert _{p}=1$. Notice that $\beta$\ is constant that depends only on
$\mathcal{B}$.

We now assert that series (\ref{GreenFunc2}) converges in $\mathcal{D}%
_{\mathbb{C}}^{^{\ast}}\left(  \mathbb{Q}_{p}^{N}\right)  $. Indeed,
\[
\frac{1}{\left\vert \mathfrak{l}\left(  \xi\right)  \right\vert _{p}^{\alpha
}+m^{2}}=\lim_{l\rightarrow+\infty}%
{\displaystyle\sum\limits_{j=-l}^{l}}
\frac{{\LARGE 1}_{S_{j}^{N}}\left(  \xi\right)  }{\left\vert \mathfrak{l}%
\left(  \xi\right)  \right\vert _{p}^{\alpha}+m^{2}}\text{\ in \ }%
\mathcal{D}_{\mathbb{C}}^{^{\ast}}\left(  \mathbb{Q}_{p}^{N}\right)  ,
\]
now, since the Fourier transform is continuous on $\mathcal{D}_{\mathbb{C}%
}^{^{\ast}}\left(  \mathbb{Q}_{p}^{N}\right)  $, one gets%
\begin{align*}
G\left(  x;m,\alpha\right)   &  =\lim_{l\rightarrow+\infty}%
{\displaystyle\sum\limits_{j=-l}^{l}}
\mathcal{F}_{\xi\rightarrow x}^{-1}\left(  \frac{{\LARGE 1}_{S_{j}^{N}}\left(
\xi\right)  }{\left\vert \mathfrak{l}\left(  \xi\right)  \right\vert
_{p}^{\alpha}+m^{2}}\right) \\
&  =\lim_{l\rightarrow+\infty}%
{\displaystyle\sum\limits_{j=-l}^{l}}
{\LARGE g}^{\left(  j\right)  }\left(  x;m,\alpha\right)  =%
{\displaystyle\sum\limits_{l=-\infty}^{+\infty}}
{\LARGE g}^{\left(  l\right)  }\left(  x;m,\alpha\right)  .
\end{align*}
Consider now $x\neq0$, with $\left\Vert x\right\Vert _{p}=p^{k}$,
$k\in\mathbb{Z}$. By changing variables as $\xi=p^{-l}z$ in ${\LARGE g}%
^{\left(  l\right)  }\left(  x;m,\alpha\right)  $, one gets%
\[
{\LARGE g}^{\left(  l\right)  }\left(  x;m,\alpha\right)  =p^{lN}%
{\displaystyle\int\limits_{S_{0}^{N}}}
\frac{\chi_{p}\left(  -p^{-l}\mathcal{B}\left(  x,z\right)  \right)
}{p^{ld\alpha}\left\vert \mathfrak{l}\left(  z\right)  \right\vert
_{p}^{\alpha}+m^{2}}d^{N}z.
\]
There exists a covering of $S_{0}^{N}$\ of the form%
\begin{equation}
S_{0}^{N}=%
{\displaystyle\bigsqcup\limits_{i=1}^{M}}
\text{ }\widetilde{z}_{i}+\left(  p^{L}\mathbb{Z}_{p}\right)  ^{n}
\label{covering1}%
\end{equation}
with $\widetilde{z}_{i}\in S_{0}^{N}$ for $i=1,\ldots,M$, and $L\in
\mathbb{N\smallsetminus}\left\{  0\right\}  $, such that%
\begin{equation}
\left\vert \mathfrak{l}\left(  z\right)  \right\vert _{p}^{\alpha}%
{\LARGE \mid}_{\widetilde{z}_{i}+\left(  p^{L}\mathbb{Z}_{p}\right)  ^{n}%
}=\left\vert \mathfrak{l}\left(  \widetilde{z}_{i}\right)  \right\vert
_{p}^{\alpha}\text{ \ for \ }i=1,\ldots,M, \label{covering2}%
\end{equation}
cf. \cite[Lemma 3]{Zu1}. From this fact and (\ref{Formula_B}), one gets
\[
{\LARGE g}^{\left(  l\right)  }\left(  x;m,\alpha\right)  =p^{lN-LN}%
{\displaystyle\sum\limits_{i=1}^{M}}
\frac{\chi_{p}\left(  -p^{-l}\mathcal{B}\left(  x,\widetilde{z}_{i}\right)
\right)  }{p^{ld\alpha}\left\vert \mathfrak{l}\left(  \widetilde{z}%
_{i}\right)  \right\vert _{p}^{\alpha}+m^{2}}%
{\displaystyle\int\limits_{\mathbb{Z}_{p}^{N}}}
\chi_{p}\left(  -p^{-l+L-k+\beta}%
{\displaystyle\sum\limits_{i=1}^{N}}
y_{i}\widetilde{A}_{i}\right)  d^{N}y.
\]
Notice that ${\LARGE g}^{\left(  l\right)  }\left(  x;m,\alpha\right)  $ is
locally constant for $\left\Vert x\right\Vert _{p}=p^{k}$. We now recall that%
\begin{equation}%
{\displaystyle\int\limits_{\mathbb{Z}_{p}^{N}}}
\chi_{p}\left(  -p^{-l+L-k+\beta}%
{\displaystyle\sum\limits_{i=1}^{N}}
y_{i}\widetilde{A}_{i}\right)  d^{N}y=\left\{
\begin{array}
[c]{lll}%
1 & \text{if} & l\leq L-k+\beta\\
&  & \\
0 & \text{if} & l\geq L-k+\beta+1.
\end{array}
\right.  \label{formula2}%
\end{equation}
Hence,
\begin{equation}
{\LARGE g}^{\left(  l\right)  }\left(  x;m,\alpha\right)  =0\text{ \ if
\ }l\geq1+L-k+\beta\label{e_meqzero}%
\end{equation}
and%
\begin{equation}
\left\vert {\LARGE g}^{\left(  l\right)  }\left(  x;m,\alpha\right)
\right\vert \leq\frac{p^{lN-LN}M}{p^{ld\alpha}\gamma+m^{2}}\text{, if }l\leq
L-k+\beta\text{,} \label{eq1}%
\end{equation}
where $\gamma=\min_{i}\left\vert \mathfrak{l}\left(  \widetilde{z}_{i}\right)
\right\vert _{p}^{\alpha}>0$, and since%
\[
Mp^{-LN}%
{\displaystyle\sum\limits_{l=-\infty}^{L-k+\beta}}
\frac{p^{lN}}{p^{ld\alpha}\gamma+m^{2}}<\frac{Mp^{-LN}}{m^{2}}%
{\displaystyle\sum\limits_{l=-\infty}^{L-k+\beta}}
p^{lN}<\infty,
\]
we have that series (\ref{GreenFunc2}) converges uniformly on the sphere
$\left\Vert x\right\Vert _{p}=p^{k}$, and equivalently (\ref{GreenFunc2})
converges uniformly on compact subsets of $\mathbb{Q}_{p}^{N}\smallsetminus
\left\{  0\right\}  $. Therefore, $G\left(  x;m,\alpha\right)  $ is a
continuous function on $\mathbb{Q}_{p}^{N}\smallsetminus\left\{  0\right\}  $.

(ii) If $N<\alpha d$, estimate (\ref{eq1}) implies the uniform convergence on
$\mathbb{Q}_{p}^{N}$ and the continuity of $G\left(  x;m,\alpha\right)  $ on
the whole $\mathbb{Q}_{p}^{N}$.

(iii) If $0<\alpha d<N$, then by (\ref{eq1}) for $\left\Vert x\right\Vert
_{p}\leq1$,%
\begin{multline*}
\left\vert G\left(  x;m,\alpha\right)  \right\vert \leq C%
{\displaystyle\sum\limits_{l=-\infty}^{L-k+\beta}}
\frac{p^{lN}}{p^{ld\alpha}\gamma+m^{2}}\leq\frac{C}{\gamma}%
{\displaystyle\sum\limits_{l=-\infty}^{L-k+\beta}}
p^{l\left(  N-d\alpha\right)  }\\
=\frac{C}{\gamma}p^{\left(  L-k+\beta\right)  \left(  N-\alpha d\right)  }%
{\displaystyle\sum\limits_{j=0}^{\infty}}
p^{-j\left(  N-\alpha d\right)  }\leq\frac{C}{\gamma}\frac{p^{\left(
L+\beta\right)  \left(  N-\alpha d\right)  }}{1-p^{\alpha d-N}}\left\Vert
x\right\Vert _{p}^{\alpha d-N}\text{.}%
\end{multline*}
If $\alpha d=N$, then for $\left\Vert x\right\Vert _{p}\leq1$, from
(\ref{eq1}),
\[
\left\vert G\left(  x;m,\alpha\right)  \right\vert \leq\frac{C}{m^{2}}%
{\displaystyle\sum\limits_{l=-\infty}^{0}}
p^{lN}+\frac{C}{\gamma}%
{\displaystyle\sum\limits_{l=1}^{L-k+\beta}}
1=C_{0}-C_{1}\ln\left\Vert x\right\Vert _{p}\text{.}%
\]

(iv) Let $\left\Vert x\right\Vert _{p}=p^{k}$. We notice that if $\left\Vert
\xi\right\Vert _{p}\leq p^{-k+\beta}$, then $\chi_{p}\left(  -\mathcal{B}%
\left(  x,\xi\right)  \right)  =1$. Therefore%
\[
G\left(  x;m,\alpha\right)  =G^{\left(  1\right)  }\left(  x;m,\alpha\right)
+G^{\left(  2\right)  }\left(  x;m,\alpha\right)  ,
\]
where%
\[
G^{\left(  1\right)  }\left(  x;m,\alpha\right)  =%
{\displaystyle\sum\limits_{l=-\infty}^{-k+\beta}}
p^{lN}%
{\displaystyle\int\limits_{S_{0}^{N}}}
\frac{d^{N}z}{p^{ld\alpha}\left\vert \mathfrak{l}\left(  z\right)  \right\vert
_{p}^{\alpha}+m^{2}},
\]%
\[
G^{\left(  2\right)  }\left(  x;m,\alpha\right)  =%
{\displaystyle\sum\limits_{l=-k+\beta+1}^{L-k+\beta}}
p^{lN}%
{\displaystyle\int\limits_{S_{0}^{N}}}
\frac{\chi_{p}\left(  -p^{-l}\mathcal{B}\left(  x,z\right)  \right)
}{p^{ld\alpha}\left\vert \mathfrak{l}\left(  z\right)  \right\vert
_{p}^{\alpha}+m^{2}}d^{N}z.
\]
By using the formula%
\begin{equation}
\frac{1}{m^{2}+t}=m^{-2}-m^{-4}t+O\left(  t^{2}\right)  \text{, }%
t\rightarrow0\text{,} \label{formula}%
\end{equation}
we get
\[
p^{lN}%
{\displaystyle\int\limits_{S_{0}^{N}}}
\frac{d^{N}z}{p^{ld\alpha}\left\vert \mathfrak{l}\left(  z\right)  \right\vert
_{p}^{\alpha}+m^{2}}=p^{lN}\left(  1-p^{-N}\right)  m^{-2}-p^{l\left(
N+d\alpha\right)  }m^{-4}Z(\alpha)+O(p^{l\left(  N+2d\alpha\right)  })\text{,}%
\]
as $t\rightarrow0$, where $Z(\alpha):=\int_{S_{0}^{N}}\left\vert
\mathfrak{l}\left(  z\right)  \right\vert _{p}^{\alpha}d^{N}z$, hence%
\begin{multline*}
G^{\left(  1\right)  }\left(  x;m,\alpha\right)  =\left(  1-p^{-N}\right)
m^{-2}%
{\displaystyle\sum\limits_{l=-\infty}^{-k+\beta}}
p^{lN}-Z(\alpha)m^{-4}%
{\displaystyle\sum\limits_{l=-\infty}^{-k+\beta}}
p^{l\left(  N+d\alpha\right)  }\\
+O(%
{\displaystyle\sum\limits_{l=-\infty}^{-k+\beta}}
p^{l\left(  N+2d\alpha\right)  })=p^{N\beta}m^{-2}\left\Vert x\right\Vert
_{p}^{-N}-\frac{Z(\alpha)m^{-4}p^{\left(  N+d\alpha\right)  \beta}%
}{1-p^{-N-d\alpha}}\left\Vert x\right\Vert _{p}^{-N-d\alpha}\\
+O(\left\Vert x\right\Vert _{p}^{-N-2d\alpha})\text{, as }\left\Vert
x\right\Vert _{p}\rightarrow\infty.
\end{multline*}
We now consider $G^{\left(  2\right)  }\left(  x;m,\alpha\right)  $, by using
(\ref{formula}),%
\begin{multline*}
p^{lN}%
{\displaystyle\int\limits_{S_{0}^{N}}}
\frac{\chi_{p}\left(  -p^{-l}\mathcal{B}\left(  x,z\right)  \right)
}{p^{ld\alpha}\left\vert \mathfrak{l}\left(  z\right)  \right\vert
_{p}^{\alpha}+m^{2}}d^{N}z=m^{-2}p^{lN}%
{\displaystyle\int\limits_{S_{0}^{N}}}
\chi_{p}\left(  -p^{-l}\mathcal{B}\left(  x,z\right)  \right)  d^{N}z\\
-m^{-4}p^{l\left(  d\alpha+N\right)  }%
{\displaystyle\int\limits_{S_{0}^{N}}}
\left\vert \mathfrak{l}\left(  z\right)  \right\vert _{p}^{\alpha}\chi
_{p}\left(  -p^{-l}\mathcal{B}\left(  x,z\right)  \right)  d^{N}%
z+O(p^{l\left(  N+2d\alpha\right)  }).
\end{multline*}
By using
\begin{align*}%
{\displaystyle\int\limits_{S_{0}^{N}}}
\chi_{p}\left(  -p^{-l}\mathcal{B}\left(  x,z\right)  \right)  d^{N}z  &  =\\%
{\displaystyle\int\limits_{S_{0}^{N}}}
\chi_{p}\left(  -p^{-l-k+\beta}%
{\displaystyle\sum\limits_{j=1}^{N}}
z_{j}\widetilde{A}_{j}\right)  d^{n}z  &  =\left\{
\begin{array}
[c]{lll}%
1-p^{-N} & \text{if} & l+k-\beta\leq0\\
&  & \\
-p^{-N} & \text{if} & l+k-\beta=1\\
&  & \\
0 & \text{if} & l+k-\beta\geq2,
\end{array}
\right.
\end{align*}
we get%
\[
m^{-2}%
{\displaystyle\sum\limits_{l=-k+\beta+1}^{L-k+\beta}}
p^{lN}%
{\displaystyle\int\limits_{S_{0}^{N}}}
\chi_{p}\left(  -p^{-l}\mathcal{B}\left(  x,z\right)  \right)  d^{N}%
z=-m^{-2}p^{N\beta}\left\Vert x\right\Vert _{p}^{-N},
\]
and by using (\ref{covering1})-(\ref{covering2})%
\begin{multline*}
-m^{-4}p^{l\left(  d\alpha+N\right)  }%
{\displaystyle\int\limits_{S_{0}^{N}}}
\left\vert \mathfrak{l}\left(  z\right)  \right\vert _{p}^{\alpha}\chi
_{p}\left(  -p^{-l}\mathcal{B}\left(  x,z\right)  \right)  d^{N}z=\\
-m^{-4}p^{l\left(  d\alpha+N\right)  -LN}%
{\displaystyle\sum\limits_{i=1}^{M}}
\left\vert \mathfrak{l}\left(  \widetilde{z}_{i}\right)  \right\vert
_{p}^{\alpha}\chi_{p}\left(  -p^{-l}\mathcal{B}\left(  x,\widetilde{z}%
_{i}\right)  \right)
{\displaystyle\int\limits_{\mathbb{Z}_{p}^{N}}}
\chi_{p}\left(  -p^{-l+L-k}\mathcal{B}\left(  x_{0},y\right)  \right)  d^{N}y,
\end{multline*}
where $x=p^{-k}x_{0}$, $\left\Vert x_{0}\right\Vert _{p}=1$. By using
(\ref{formula2})%
\begin{multline*}
-m^{-4}p^{-LN}%
{\displaystyle\sum\limits_{i=1}^{M}}
\left\vert \mathfrak{l}\left(  \widetilde{z}_{i}\right)  \right\vert
_{p}^{\alpha}%
{\displaystyle\sum\limits_{l=-k+\beta+1}^{L-k+\beta}}
\chi_{p}\left(  -p^{-l}\mathcal{B}\left(  x_{0},\widetilde{z}_{i}\right)
\right)  p^{l\left(  d\alpha+N\right)  }\\
=-m^{-4}p^{d\alpha+N-LN}\left(
{\displaystyle\sum\limits_{i=1}^{M}}
\left\vert \mathfrak{l}\left(  \widetilde{z}_{i}\right)  \right\vert
_{p}^{\alpha}%
{\textstyle\sum\limits_{j=0}^{L-1}}
\chi_{p}\left(  -p^{j-k+\beta+1}\mathcal{B}\left(  x,\widetilde{z}_{i}\right)
\right)  p^{j\left(  d\alpha+N\right)  }\right)  \times\\
p^{\beta\left(  d\alpha+N\right)  }\left\Vert x\right\Vert _{p}^{-d\alpha-N},
\end{multline*}
therefore%
\begin{multline*}
G\left(  x;m,\alpha\right)  =G^{\left(  1\right)  }\left(  x;m,\alpha\right)
+G^{\left(  2\right)  }\left(  x;m,\alpha\right)  =-m^{-4}A\left(  x\right)
p^{\beta\left(  d\alpha+N\right)  }\left\Vert x\right\Vert _{p}^{-d\alpha-N}\\
+O(\left\Vert x\right\Vert _{p}^{-N-2d\alpha})\text{, as }\left\Vert
x\right\Vert _{p}\rightarrow\infty,
\end{multline*}
where
\begin{multline*}
A\left(  x\right)  :=\frac{Z(\alpha)}{1-p^{-N-d\alpha}}\\
+p^{d\alpha+N-LN}\left(
{\displaystyle\sum\limits_{i=1}^{M}}
\left\vert \mathfrak{l}\left(  \widetilde{z}_{i}\right)  \right\vert
_{p}^{\alpha}%
{\textstyle\sum\limits_{j=0}^{L-1}}
\chi_{p}\left(  -p^{j-k+\beta+1}\mathcal{B}\left(  x,\widetilde{z}_{i}\right)
\right)  p^{j\left(  d\alpha+N\right)  }\right)  .
\end{multline*}

(iv) By (i) we have that $G\left(  x;m,\alpha\right)  $ is a continuous
function on $\mathbb{Q}_{p}^{n}\smallsetminus\left\{  0\right\}  $ having
expansion (\ref{GreenFunc2}). Thus, it is sufficient to show that
\begin{equation}
{\LARGE g}^{\left(  l\right)  }\left(  x;m,\alpha\right)  \geq0\text{ on
}\mathbb{Q}_{p}^{n}\smallsetminus\left\{  0\right\}  \text{.}
\label{condition}%
\end{equation}
Take $t\in\mathbb{Q}_{p}^{\times}$, then $V_{t}:=\left\{  \xi\in S_{j}%
^{N}:\mathfrak{l}\left(  \xi\right)  =t\right\}  $ is a $p$-adic compact
submanifold of the sphere $S_{j}^{N}$. Then, there exists a differential form
$\omega_{0}$, the Gel'fand-Leray form, such that $d\xi_{1}\wedge\cdots\wedge
d\xi_{n}=\omega_{0}\wedge d\mathfrak{l}$. Denote the measure corresponding to
$\omega_{0}$\ as $\frac{d\xi}{d\mathfrak{l}}$, then%
\[
{\LARGE g}^{\left(  l\right)  }\left(  x;m,\alpha\right)  =%
{\displaystyle\int\limits_{\mathbb{Q}_{p}^{\times}}}
\frac{1}{\left\vert t\right\vert _{p}^{\alpha}+m^{2}}\left\{
{\displaystyle\int\limits_{V_{t}}}
\chi_{p}\left(  -\mathcal{B}\left(  x,\xi\right)  \right)  \frac{d\xi
}{d\mathfrak{l}}\right\}  dt.
\]
Thus, in order to establish (\ref{condition}), it is sufficient to show that
\[%
{\displaystyle\int\limits_{V_{t}}}
\chi_{p}\left(  -\mathcal{B}\left(  x,\xi\right)  \right)  \frac{d\xi
}{d\mathfrak{l}}=%
{\displaystyle\int\limits_{V_{t}}}
\chi_{p}\left(  -p^{ord\left(  x\right)  +\beta}%
{\displaystyle\sum\limits_{i=1}^{N}}
\xi_{i}\widetilde{A}_{i}\right)  \frac{d\xi}{d\mathfrak{l}}\geq0\text{ for all
}x\text{.}%
\]
This last inequality is established as in the proof of Theorem 2 in
\cite{Zu1}, by using the non-Archimedean Implicit Function Theorem and by
performing a suitable change of variables.
\end{proof}

\subsection{Green functions on $\mathcal{H}_{\mathbb{R}}\left(  \infty\right)
$}

\begin{theorem}
\label{Thm3}Let $\alpha>0$, $m>0$, and let $\boldsymbol{L}_{\alpha}$\ be an
elliptic operator. (i) There exists a Green function $G\left(  x;m,\alpha
\right)  $ for the operator $\boldsymbol{L}_{\alpha}$, which is continuous and
non-negative on $\mathbb{Q}_{p}^{n}\smallsetminus\left\{  0\right\}  $, and
tends to zero at infinity. Furthermore, if $h\in\mathcal{D}_{\mathbb{C}%
}(\mathbb{Q}_{p}^{N})$, then $u(x)=G\left(  x;m,\alpha\right)  \ast h(x)$ is a
solution of (\ref{equation1}) in $\mathcal{D}_{\mathbb{C}}^{^{\ast}%
}(\mathbb{Q}_{p}^{N})$. (ii) The equation%
\begin{equation}
\left(  \boldsymbol{L}_{\alpha}+m^{2}\right)  u=g\text{, } \label{equation3A}%
\end{equation}
\ with $g\in\mathcal{H}_{\mathbb{R}}\left(  \infty\right)  $, has a unique
solution $u\in\mathcal{H}_{\mathbb{R}}\left(  \infty\right)  $.
\end{theorem}

\begin{proof}
(i) We first notice that for any $h\in\mathcal{D}_{\mathbb{C}}(\mathbb{Q}%
_{p}^{N})$, $u(x):=G\left(  x;m,\alpha\right)  \ast h(x)$ is a locally
constant function because it is the inverse Fourier transform of
$\frac{\widehat{g}\left(  \xi\right)  }{\left\vert \mathfrak{l}\left(
\xi\right)  \right\vert _{p}^{\alpha}+m^{2}}$, which is a distribution with
compact support. Taking $u(x)=G\left(  x;m,\alpha\right)  \ast h(x)\in
\mathcal{D}_{\mathbb{C}}^{^{\ast}}(\mathbb{Q}_{p}^{N})$, we have $\left(
\boldsymbol{L}_{\alpha}+m^{2}\right)  u=h$\ in $\mathcal{D}_{\mathbb{C}%
}^{^{\ast}}(\mathbb{Q}_{p}^{N})$. The other results follow from Proposition
\ref{Prop1}.

(ii) Take $g\in\mathcal{D}_{\mathbb{R}}(\mathbb{Q}_{p}^{N})$, then by (i),
$u(x)=G\left(  x;m,\alpha\right)  \ast g(x)$ is a real-valued, locally
constant function which is a solution of (\ref{equation3A}) in $\mathcal{D}%
_{\mathbb{C}}^{^{\ast}}(\mathbb{Q}_{p}^{N})$. Now, since $\widehat{u}\left(
\xi\right)  =\frac{\widehat{g}\left(  \xi\right)  }{\left\vert \mathfrak{l}%
\left(  \xi\right)  \right\vert _{p}^{\alpha}+m^{2}}\in L^{2}$, by using
(\ref{basic_estimate}),%
\begin{align*}
\left\Vert u\right\Vert _{l+d}^{2}  &  \leq C\left\Vert g\right\Vert _{0}^{2}+%
{\displaystyle\int\limits_{\mathbb{Q}_{p}^{N}\smallsetminus B_{0}^{N}}}
\frac{\left\Vert \xi\right\Vert _{p}^{2\alpha\left(  l+d\right)  }\left\vert
\widehat{g}\left(  \xi\right)  \right\vert ^{2}d^{N}\xi}{\left(  \left\vert
\mathfrak{l}\left(  \xi\right)  \right\vert _{p}^{\alpha}+m^{2}\right)  ^{2}%
}\\
&  \leq C\left\Vert g\right\Vert _{0}^{2}+\frac{1}{C_{0}^{\alpha}}%
{\displaystyle\int\limits_{\mathbb{Q}_{p}^{N}\smallsetminus B_{0}^{N}}}
\left\Vert \xi\right\Vert _{p}^{2\alpha l}\left\vert \widehat{g}\left(
\xi\right)  \right\vert ^{2}d^{N}\xi\leq C\left\Vert g\right\Vert _{0}%
^{2}+\frac{1}{C_{0}^{\alpha}}\left\Vert g\right\Vert _{l}^{2}\\
&  \leq C^{\prime}\left\Vert g\right\Vert _{l}^{2}\text{, for }l\in\mathbb{N}.
\end{align*}
Then, by Lemma \ref{lemma2A}, $u\in\mathcal{H}_{\mathbb{R}}\left(  m\right)
$, for $m\geq d$. In the case, $0\leq m\leq d-1$, one gets $\left\Vert
u\right\Vert _{m}\leq C^{\prime\prime}\left\Vert g\right\Vert _{0}$. Therefore
$u\in\mathcal{H}_{\mathbb{R}}\left(  m\right)  $, for $m\in\mathbb{N}$. We now
take a sequence $g_{n}\in\mathcal{D}_{\mathbb{R}}(\mathbb{Q}_{p}^{N})$ such
that $g_{n}$ $\underrightarrow{d}$ $g\in\mathcal{H}_{\mathbb{R}}\left(
\infty\right)  $ and $\left(  \boldsymbol{L}_{\alpha}+m^{2}\right)
u_{n}=g_{n}$ with $u_{n}\in\mathcal{H}_{\mathbb{R}}\left(  \infty\right)  $.
\ By the continuity of $\left(  \boldsymbol{L}_{\alpha}+m^{2}\right)
^{-1}\mid_{\mathcal{H}_{\mathbb{R}}\left(  \infty\right)  }$ and the density
\ of $\mathcal{D}_{\mathbb{R}}(\mathbb{Q}_{p}^{N})$ in $\mathcal{H}%
_{\mathbb{R}}\left(  \infty\right)  $, we have $u_{n}$ $\underrightarrow{d}$
$u\in\mathcal{H}_{\mathbb{R}}\left(  \infty\right)  $ and $\left(
\boldsymbol{L}_{\alpha}+m^{2}\right)  u=g$, cf. Theorem \ref{Thm2}. The
uniqueness of $u$ follows from Theorem \ref{Thm2}.
\end{proof}

\begin{corollary}
\label{Cor4}The mapping%
\[%
\begin{array}
[c]{ccc}%
\mathcal{H}_{\mathbb{R}}\left(  \infty\right)  & \rightarrow & \mathcal{H}%
_{\mathbb{R}}\left(  \infty\right) \\
g\left(  x\right)  & \rightarrow & G\left(  x;m,\alpha\right)  \ast g(x),
\end{array}
\]
is continuous.
\end{corollary}

\begin{proof}
By the proof of Theorem \ref{Thm3} (ii), and Theorem \ref{Thm2} ,
\[%
\begin{array}
[c]{ccc}%
\mathcal{H}_{\mathbb{R}}\left(  \infty\right)  & \rightarrow & \mathcal{H}%
_{\mathbb{R}}\left(  \infty\right) \\
g & \rightarrow & \left(  \boldsymbol{L}_{\alpha}+m^{2}\right)  ^{-1}g,
\end{array}
\]

is a well-defined continuous mapping.
\end{proof}

\begin{remark}
(i) It is worth to mention that the Archimedean and non-Archimedean Green
functions share similar properties, cf. \cite[Proposition 7.2.1]{Glimm-Jaffe}
and Proposition \ref{Prop1} and Theorem \ref{Thm3}.

(ii) Proposition \ref{Prop1} and Theorem \ref{Thm3} generalize to arbitrary
dimension some results established by Kochubei for Klein-Gordon
pseudodifferential operators attached with elliptic quadratic forms cf.
\cite[Proposition 2.8 and Theorem 2.4]{Koch}.

(iii) Another possible definition for the $p$-adic Klein-Gordon operator, with
real mass $m>0$, is the following:%
\[
\varphi\rightarrow\mathcal{F}^{-1}\left(  \left(  \left\vert \mathfrak{l}%
\right\vert _{p}+m^{2}\right)  ^{\alpha}\mathcal{F}\varphi\right)  .
\]
But, since Proposition \ref{Prop1} and Theorem \ref{Thm3} \ remain valid for
this type of operators, we prefer to work with (\ref{equation1}) because this
type of operators have been studied extensively in the $p$-adic setting.

(iv) Theorem \ref{Thm3} shows that $\mathcal{H}_{\mathbb{C}}\left(
\infty\right)  $ contains distributions.
\end{remark}

\section{\label{Sect5}The Generalized White Noise}

\subsection{Infinitely divisible probability distributions}

We recall that an infinitely divisible probability distribution $P$ is a
probability distribution having the property that for each $n\in\mathbb{N}$
there exists a probability distribution $P_{n}$ such that $P=P_{n}\ast
\cdots\ast P_{n}$ ($n$-times). By the Lévy-Khinchine Theorem, the
characteristic function $C_{P}$ of $P$ satisfies%
\begin{equation}
C_{P}(t)=%
{\textstyle\int\limits_{\mathbb{R}}}
e^{ist}dP(s)=e^{\Psi\left(  t\right)  }\text{, }t\in\mathbb{R}\text{,}
\label{char_function}%
\end{equation}
where $\Psi:\mathbb{R}\rightarrow\mathbb{C}$ is a continuous function, called
the \textit{Lévy characteristic of} $P$, which is uniquely represented as
follows:%
\begin{equation}
\Psi\left(  t\right)  =iat-\frac{\sigma^{2}t^{2}}{2}+%
{\textstyle\int\limits_{\mathbb{R\smallsetminus}\left\{  0\right\}  }}
\left(  e^{ist}-1-\frac{ist}{1+s^{2}}\right)  dM(s)\text{, }t\in
\mathbb{R}\text{,} \label{Psi}%
\end{equation}
where $a$, $\sigma\in\mathbb{R}$, with $\sigma\geq0$, and the measure $dM(s)$
satisfies%
\begin{equation}%
{\textstyle\int\limits_{\mathbb{R\smallsetminus}\left\{  0\right\}  }}
\min\left(  1,s^{2}\right)  dM(s)<\infty. \label{dM(s)}%
\end{equation}
On the other hand, given a triple $\left(  a,\sigma,dM\right)  $\ with
$a\in\mathbb{R}$, $\sigma\geq0$, and $dM$ a measure on
$\mathbb{R\smallsetminus}\left\{  0\right\}  $ satisfying (\ref{dM(s)}), there
exists a unique infinitely divisible probability distribution $P$ such that
its Lévy characteristic is given by (\ref{Psi}).

\begin{remark}
From now on, we work with infinitely divisible probability distributions which
are absolutely continuous with all finite moments. This fact is equivalent to
all the moments of the corresponding $M$'s are finite, cf. \cite[Theorem
2.3]{Albeverio-et-al-4}.
\end{remark}

Let $N\in\mathbb{N}$\ be as before. Let $\mathcal{H}_{\mathbb{R}}\left(
\infty\right)  $ and $\mathcal{H}_{\mathbb{R}}^{^{\ast}}\left(  \infty\right)
$ be the spaces introduced in Section \ref{Nuclear_spaces}. We denote by
$\left\langle \cdot,\cdot\right\rangle $ the dual pairing between
$\mathcal{H}_{\mathbb{R}}\left(  \infty\right)  $ and $\mathcal{H}%
_{\mathbb{R}}^{^{\ast}}\left(  \infty\right)  $. Let $\mathcal{B}$ be the
$\sigma$-algebra generated by cylinder sets of $\mathcal{H}_{\mathbb{R}%
}^{^{\ast}}\left(  \infty\right)  $. Then $\left(  \mathcal{H}_{\mathbb{R}%
}^{^{\ast}}\left(  \infty\right)  ,\mathcal{B}\right)  $ is a measurable space.

By a \textit{characteristic functional} on $\mathcal{H}_{\mathbb{R}}\left(
\infty\right)  $, we mean a functional $C:\mathcal{H}_{\mathbb{R}}\left(
\infty\right)  \rightarrow\mathbb{C}$ satisfying the following properties:

\noindent(i) $C$ is continuous on $\mathcal{H}_{\mathbb{R}}\left(
\infty\right)  $;

\noindent(ii) $C$ is positive-definite;

\noindent(iii) $C(0)=1$.

Now, since $\mathcal{H}_{\mathbb{R}}\left(  \infty\right)  $ is a nuclear
space, cf. Theorem \ref{Thm1}, by the Bochner-Minlos Theorem (see e.g.
\cite{Minlos}), there exists a one to one correspondence between the
characteristic functionals $C$ and probability measures $\mathrm{P}$ on
$\left(  \mathcal{H}_{\mathbb{R}}^{^{\ast}}\left(  \infty\right)
,\mathcal{B}\right)  $ given by the following relation%
\[
C(f)=%
{\textstyle\int\limits_{\mathcal{H}_{\mathbb{R}}^{^{\ast}}\left(
\infty\right)  }}
e^{i\left\langle f,T\right\rangle }\mathrm{dP}\left(  T\right)  \text{, }%
f\in\mathcal{H}_{\mathbb{R}}\left(  \infty\right)  .
\]

\begin{theorem}
\label{Thm5}Let $\Psi$ be \ a Lévy characteristic defined by
(\ref{char_function}). Then there exists a unique probability measure
\textrm{P}$_{\Psi}$ on $\left(  \mathcal{H}_{\mathbb{R}}^{^{\ast}}\left(
\infty\right)  ,\mathcal{B}\right)  $ such that the Fourier transform of
\textrm{P}$_{\Psi}$ satisfies%
\[%
{\textstyle\int\limits_{\mathcal{H}_{\mathbb{R}}^{^{\ast}}\left(
\infty\right)  }}
e^{i\left\langle f,T\right\rangle }\mathrm{dP}_{\Psi}\left(  T\right)
=\exp\left\{
{\textstyle\int\limits_{\mathbb{Q}_{p}^{N}}}
\Psi\left(  f\left(  x\right)  \right)  d^{N}x\right\}  \text{, }%
f\in\mathcal{H}_{\mathbb{R}}\left(  \infty\right)  .
\]

\end{theorem}

The proof is based on \cite[Theorem 6, p. 283]{Gel-Vil} like in the
Archimedean case, cf. \cite[Theorem 1.1]{Albeverio-et-al-4}. However, in the
non-Archimedean case the result does not follow directly from \cite{Gel-Vil}.
We need some additional results.

\begin{lemma}
\label{lemma6A}$\int_{\mathbb{Q}_{p}^{N}}\Psi\left(  f\left(  x\right)
\right)  d^{N}x<\infty$ for any $f\in\mathcal{H}_{\mathbb{R}}\left(
\infty\right)  $.
\end{lemma}

\begin{proof}
By formula (\ref{Psi}), we have to show that : (i) $\int_{\mathbb{Q}_{p}^{N}%
}f\left(  x\right)  d^{N}x<\infty$; (ii) $\int_{\mathbb{Q}_{p}^{N}}%
f^{2}\left(  x\right)  d^{N}x<\infty$; (iii) $\int_{\mathbb{Q}_{p}^{N}}%
\int_{\mathbb{R\smallsetminus}\left\{  0\right\}  }\left(  e^{isf\left(
x\right)  }-1-\frac{isf\left(  x\right)  }{1+s^{2}}\right)  dM(s)d^{N}%
x<\infty$ for any $f\in\mathcal{H}_{\mathbb{R}}\left(  \infty\right)  $. (i),
(ii) follow from the fact that $f\in L^{1}$, respectively that $f\in L^{2}$,
cf. Theorem \ref{Thm2}\ (ii). To verify (iii) we use that $\left\vert
e^{isf\left(  x\right)  }-1\right\vert \leq\left\vert sf\left(  x\right)
\right\vert $, $s\in\mathbb{R}$, $x\in\mathbb{Q}_{p}^{N}$, then integral (iii)
is bounded by $2\left(  \int_{\mathbb{R\smallsetminus}\left\{  0\right\}
}\left\vert s\right\vert dM(s)\right)  \left(  \int_{\mathbb{Q}_{p}^{N}%
}\left\vert f\left(  x\right)  \right\vert d^{N}x\right)  $ which is finite
because $M$ has finite moments and $f\in L^{1}$.
\end{proof}

\begin{lemma}
\label{lemma6}The function $f\rightarrow$ $\int_{\mathbb{Q}_{p}^{N}}%
\Psi\left(  f\left(  x\right)  \right)  d^{N}x$\ is continuous on
$\mathcal{H}_{\mathbb{R}}\left(  \infty\right)  $.
\end{lemma}

\begin{proof}
By (\ref{Psi}) and Lemma \ref{lemma2} (i), it is sufficient to show that
\begin{equation}
f\rightarrow D(f):=%
{\textstyle\int\limits_{\mathbb{Q}_{p}^{N}}}
\text{ }%
{\textstyle\int\limits_{\mathbb{R\smallsetminus}\left\{  0\right\}  }}
\left(  e^{isf\left(  x\right)  }-1-\frac{isf\left(  x\right)  }{1+s^{2}%
}\right)  dM(s)d^{N}x \label{equation4}%
\end{equation}
is continuous on $\left(  \mathcal{D}_{\mathbb{R}}(\mathbb{Q}_{p}%
^{N}),d\right)  $. Take $f\in\mathcal{D}_{\mathbb{R}}(\mathbb{Q}_{p}^{N})$ and
a sequence $\left\{  f_{n}\right\}  _{n\in\mathbb{N}}$ in $\mathcal{D}%
_{\mathbb{R}}(\mathbb{Q}_{p}^{N})$ such that $f_{n}$ $\underrightarrow{d}$
$f$, i.e. $f_{n}$ $\underrightarrow{\left\Vert \cdot\right\Vert _{m}}$ $f$
$\ $for every $m$. By contradiction assume that $f-D(f)$ does not converge to
$0$, then there exist $\epsilon>0$ and a subsequence $g_{k}=f_{n_{k}}$ such
that $\left\vert g_{k}-D(f)\right\vert >\epsilon$. On the other hand, taking
$m=0$, we have $g_{k}$ $\underrightarrow{L^{2}}$ $f$, and thus there is a
subsequence $\left\{  g_{k_{j}}\right\}  _{k_{j}}$ such that $g_{k_{j}}$
$\rightarrow$ $f$ almost uniformly. Now, since the support of $f$ is contained
in a ball, say $B_{N_{0}}^{n}$, then the support of each\ $g_{k_{j}}$ is
contained in $B_{N_{0}}^{n}$ almost everywhere. Hence $\left\vert g_{k_{j}%
}\left(  x\right)  \right\vert \leq C_{0}\left\Vert f\right\Vert _{\infty
}1_{B_{N0}^{n}}\left(  x\right)  $ almost everywhere, for some positive
constant $C_{0}$. By using that $\left\vert e^{is\left(  f\left(  x\right)
-g_{k_{j}}\left(  x\right)  \right)  }-1\right\vert \leq\left\vert
s\right\vert \left\vert f\left(  x\right)  -g_{k_{j}}\left(  x\right)
\right\vert $, $s\in\mathbb{R}$, $x\in\mathbb{Q}_{p}^{N}$, and that $M$ is a
bounded measure with finite moments,
\begin{align*}
\left\vert g_{k_{j}}-D(f)\right\vert  &  \leq2\left\{
{\textstyle\int\limits_{\mathbb{R\smallsetminus}\left\{  0\right\}  }}
s^{2}dM(s)\right\}
{\textstyle\int\limits_{\mathbb{Q}_{p}^{N}}}
\left\vert f\left(  x\right)  -g_{k_{j}}\left(  x\right)  \right\vert
d^{N}x\text{ }\\
&  \leq C%
{\textstyle\int\limits_{\mathbb{Q}_{p}^{N}}}
\left\vert f\left(  x\right)  -g_{k_{j}}\left(  x\right)  \right\vert d^{N}x,
\end{align*}
\ by Dominated Convergence Theorem, using that $f\in L^{1}$, and that
$\left\vert g_{k_{j}}\left(  x\right)  \right\vert \leq C\left\Vert
f\right\Vert _{\infty}1_{B_{N0}^{n}}\left(  x\right)  $ almost everywhere,
i.e. $\left\vert g_{k_{j}}\left(  x\right)  \right\vert \leq\left\vert
g\left(  x\right)  \right\vert \in L^{1}$, we conclude that $D\left(
g_{k_{j}}\right)  $ $\rightarrow D$ $\left(  f\right)  $, which contradicts
$\left\vert g_{k}-D(f)\right\vert >\epsilon$.
\end{proof}

Set $L(f):=\exp\left\{  \int_{\mathbb{Q}_{p}^{N}}\Psi\left(  f\left(
x\right)  \right)  d^{N}x\right\}  $ for $f\in\mathcal{H}_{\mathbb{R}}\left(
\infty\right)  $. Notice that by Lemma \ref{lemma6A} this function is well-defined.

\begin{proposition}
\label{Prop2}The function $L(f)$ is positive-definite if and only if
$e^{s\Psi\left(  t\right)  }$\ is positive-definite for every $s>0$.
\end{proposition}

\begin{proof}
Suppose that $L(f)$ is positive-definite, i.e.
\begin{equation}%
{\textstyle\sum\limits_{j,k=1}^{m}}
L\left(  f_{j}-f_{k}\right)  z_{j}\overline{z}_{k}\geq0\text{ for }%
f_{j}\text{, }f_{k}\in\mathcal{H}_{\mathbb{R}}\left(  \infty\right)  \text{,
}z_{j}\text{, }z_{k}\in\mathbb{C}\text{, }j\text{, }k=1,\ldots,m.
\label{equation3}%
\end{equation}

Take $f_{j}\left(  x\right)  =t_{j}\Omega\left(  p^{-l}\left\Vert
x-x_{0}\right\Vert _{p}\right)  $, $t_{j}\in\mathbb{R}$, for $j$,$=1,\ldots
,m$, $l\in\mathbb{Z}$, and $x_{0}\in\mathbb{Q}_{p}^{N}$, then%
\begin{align*}%
{\textstyle\sum\limits_{j,k=1}^{m}}
L\left(  f_{j}-f_{k}\right)  z_{j}\overline{z}_{k}  &  =%
{\textstyle\sum\limits_{j,k=1}^{m}}
\exp\left(  \int_{\left\Vert x-x_{0}\right\Vert _{p}\leq p^{l}}\Psi\left(
t_{j}-t_{k}\right)  d^{N}x\right)  z_{j}\overline{z}_{k}\\
&  =%
{\textstyle\sum\limits_{j,k=1}^{m}}
\left\{  \exp\left[  p^{Nl}\Psi\left(  t_{j}-t_{k}\right)  \right]  \right\}
z_{j}\overline{z}_{k}\geq0.
\end{align*}
This proves that $e^{p^{Nl}\Psi\left(  t\right)  }$\ is a positive-definite
function. Now, $t\rightarrow\frac{1-e^{p^{Nl}\Psi\left(  t\right)  }}{p^{Nl}}$
is negative-definite, cf. \cite[Corollary 7.7]{Berg-Forst}, for every
$l\in\mathbb{N}$. Furthermore, by \cite[Proposition 7.4 (i)]{Berg-Forst},
$\lim_{l\rightarrow\infty}\frac{1-e^{p^{Nl}\Psi\left(  t\right)  }}{p^{Nl}%
}=-\Psi\left(  t\right)  $ is a negative-definite function, and since the
negative-definite functions form a cone, $-s\Psi\left(  t\right)  $ is
negative-definite for every $s>0$. Finally, by the Schoenberg Theorem, cf.
\cite[Theorem 7.8]{Berg-Forst}, $e^{s\Psi\left(  t\right)  }$ is
positive-definite for every $s>0$.

We now assume that $e^{s\Psi\left(  t\right)  }$\ is positive-definite for
every $s>0$. In order to prove (\ref{equation3}), by Lemma \ref{lemma6}, it is
sufficient to take $f_{j}$, $f_{k}\in\mathcal{D}_{\mathbb{R}}\left(
\mathbb{Q}_{p}^{N}\right)  $, for $j$, $k=1,\ldots,m$. Consider the matrix
$A=\left[  a_{ij}\right]  $ with $a_{ij}:=\exp\left(  \int_{\mathbb{Q}_{p}%
^{N}}\Psi\left(  f_{i}\left(  x\right)  -f_{j}\left(  x\right)  \right)
d^{N}x\right)  $. We have to show that $A$ is positive-definite. Take
$B_{n}^{N}$ such that supp$f_{i}\subseteq B_{n}^{N}$ for $i=1,\ldots,m$. Then
$a_{ij}=\exp\left(  \int_{B_{n}^{N}}\Psi\left(  f_{i}\left(  x\right)
-f_{j}\left(  x\right)  \right)  d^{N}x\right)  $. By using that each
$f_{i}(x)$ is a locally constant function, and that $B_{n}^{N}$ is an open
compact set, there exists a finite covering $B_{n}^{N}=\sqcup_{l=1}%
^{L}B_{n^{\prime}}^{N}\left(  \widetilde{x}_{l}\right)  $ such that
$f_{i}(x)\mid_{B_{n^{\prime}}^{N}\left(  \widetilde{x}_{l}\right)  }%
=f_{i}(\widetilde{x}_{l})$. Hence $a_{ij}=%
{\textstyle\prod\nolimits_{l=1}^{L}}
\exp\left(  p^{Nn^{\prime}}\alpha_{l}\right)  $ with $\alpha_{l}:=\Psi\left(
f_{i}\left(  \widetilde{x}_{l}\right)  -f_{j}\left(  \widetilde{x}_{l}\right)
\right)  $. By a Schur Theorem, cf. \cite[Theorem in p. 277]{Gel-Vil}, \ A is
positive-definite if the matrix $\left[  \exp\left(  p^{Nn^{\prime}}\alpha
_{l}\right)  \right]  $ is positive-definite, which follows from the fact that
$e^{s\Psi\left(  t\right)  }$\ is positive-definite for every $s>0$.
\end{proof}

\subsubsection{Proof of Theorem \ref{Thm5}.}

By \cite[Theorem 1, p. 273]{Gel-Vil}, $L(f)\not \equiv 0$ is the
characteristic functional of a generalized random process (a random field in
our terminology) with independent values at every point, if and only if: (A)
$L$ is positive-definite and (B) for any functions $f_{1}\left(  t\right)  $,
$f_{2}\left(  t\right)  \in\mathcal{H}_{\mathbb{R}}\left(  \infty\right)  $
whose product vanishes, it verifies that $L(f_{1}+f_{2})=L(f_{1})L(f_{2})$.
The verification of condition B is straightforward. Condition A is equivalent
to $e^{s\Psi\left(  t\right)  }$\ is positive-definite for every $s>0$, cf.
Proposition \ref{Prop2}. By \cite[Theorem 4 in p. 279 and Theorem 3 in p.
189]{Gel-Vil}, this last condition turns out to be equivalent to the fact that
$\Psi$ has the form (\ref{Psi}).

\subsection{Non-Archimedean generalized white noise measures}

\begin{definition}
\label{def2}We call \textrm{P}$_{\Psi}$ in Theorem \ref{Thm5} a generalized
white noise measure with Lévy characteristic $\Psi$ and $\left(
\mathcal{H}_{\mathbb{R}}^{^{\ast}}(\infty),\mathcal{B},\mathrm{P}_{\Psi
}\right)  $ the generalized white noise space associated with $\Psi$. The
associated coordinate process%
\[
\boldsymbol{\digamma}:\mathcal{H}_{\mathbb{R}}\left(  \infty\right)
\times\left(  \mathcal{H}_{\mathbb{R}}^{^{\ast}}(\infty),\mathcal{B}%
,\mathrm{P}_{\Psi}\right)  \rightarrow\mathbb{R}%
\]
defined by $\boldsymbol{\digamma}\left(  f,T\right)  =\left\langle
f,T\right\rangle $, $f\in\mathcal{H}_{\mathbb{R}}\left(  \infty\right)  $,
$T\in\mathcal{H}_{\mathbb{R}}^{^{\ast}}(\infty)$, is called generalized white noise.
\end{definition}

The generalized white noise $\boldsymbol{\digamma}$ is composed by three
independent noises: constant, Gaussian and Poisson (with jumps given by $M$)
noises, see Remark 1.3 in \cite{Albeverio-et-al-3}.

\section{\label{Sect6}Euclidean random fields as convoluted generalized white
noise}

\subsection{Construction}

\begin{definition}
\label{def3}Let $\left(  \Omega,\mathcal{F},P\right)  $ be a given probability
space. By a generalized random field $\boldsymbol{\Phi}$ on $\left(
\Omega,\mathcal{F},P\right)  $ with parameter space $\mathcal{H}_{\mathbb{R}%
}\left(  \infty\right)  $, we mean a system
\[
\left\{  \boldsymbol{\Phi}\left(  g,\omega\right)  :\omega\in\Omega\right\}
_{g\in\mathcal{H}_{\mathbb{R}}\left(  \infty\right)  },
\]
of random variables on $\left(  \Omega,\mathcal{F},P\right)  $ having the
following properties:

\noindent(i) $P\left\{  \omega\in\Omega:\boldsymbol{\Phi}\left(  c_{1}%
g_{1}+c_{2}g_{2},\omega\right)  =c_{1}\boldsymbol{\Phi}\left(  g_{1}%
,\omega\right)  +c_{2}\boldsymbol{\Phi}\left(  g_{2},\omega\right)  \right\}
=1$, for $c_{1}$, $c_{2}\in\mathbb{R}$, $g_{1}$, $g_{2}\in\mathcal{H}%
_{\mathbb{R}}\left(  \infty\right)  $;

\noindent(ii) if $g_{n}\rightarrow g$ in $\mathcal{H}_{\mathbb{R}}\left(
\infty\right)  $, then $\boldsymbol{\Phi}\left(  g_{n},\omega\right)
\rightarrow\boldsymbol{\Phi}\left(  g,\omega\right)  $ in law.
\end{definition}

The coordinate process in Definition \ref{def2} is a random field on the
generalized white noise space $\left(  \mathcal{H}_{\mathbb{R}}^{^{\ast}%
}(\infty),\mathcal{B},\mathrm{P}_{\Psi}\right)  $, because property (i) is
fulfilled pointwise and property (ii) follows from the following fact:%
\begin{equation}
\lim_{n\rightarrow\infty}\mathrm{P}_{\Psi}\left\{  T\in\mathcal{H}%
_{\mathbb{R}}^{^{\ast}}(\infty):\left\vert \left\langle g_{n}-g,T\right\rangle
\right\vert <\epsilon\right\}  =1. \label{equation5}%
\end{equation}
Indeed, since $\mathcal{H}_{\mathbb{R}}^{^{\ast}}(\infty)$ is the union of the
increasing spaces $\mathcal{H}_{\mathbb{R}}^{^{\ast}}(l)$, there exists
$l_{0}\in\mathbb{N}$ such that $T\in\mathcal{H}_{\mathbb{R}}^{^{\ast}}(l_{0}%
)$, and thus $\left\vert \left\langle g_{n}-g,T\right\rangle \right\vert
\leq\left\Vert T\right\Vert _{-l_{0}}\left\Vert g_{n}-g\right\Vert _{l_{0}%
}\leq\left\Vert T\right\Vert _{-l_{0}}$, for $n$ big enough. Now,
(\ref{equation5}) follows by the Dominated Convergence Theorem. Therefore
$\lim_{n\rightarrow\infty}\mathrm{P}_{\Psi}\left\{  T\in\mathcal{H}%
_{\mathbb{R}}^{^{\ast}}(\infty):\left\vert \left\langle g_{n}-g,T\right\rangle
\right\vert \geq\epsilon\right\}  =0$.

We now recall that $\left(  \mathcal{G}f\right)  \left(  x\right)  :=G\left(
x;m,\alpha\right)  \ast f\left(  x\right)  $ gives rise to a continuous
mapping from $\mathcal{H}_{\mathbb{R}}(\infty)$ into itself, cf. Corollary
\ref{Cor4}. Thus, the conjugate operator $\widetilde{\mathcal{G}}%
:\mathcal{H}_{\mathbb{R}}^{^{\ast}}(\infty)\rightarrow\mathcal{H}_{\mathbb{R}%
}^{^{\ast}}(\infty)$ is a measurable mapping from $\left(  \mathcal{H}%
_{\mathbb{R}}^{^{\ast}}\left(  \infty\right)  ,\mathcal{B}\right)  $ into
itself. The generalized white noise measure $\mathrm{P}_{\Psi}$ on $\left(
\mathcal{H}_{\mathbb{R}}^{^{\ast}}\left(  \infty\right)  ,\mathcal{B}\right)
$ associated with a Lévy characteristic $\Psi$ was introduced in Definition
\ref{def2}. We set $\mathrm{P}_{\boldsymbol{\Phi}}$ to be the image
probability measure of $\mathrm{P}_{\Psi}$ under $\widetilde{\mathcal{G}}$,
i.e. $\mathrm{P}_{\boldsymbol{\Phi}}$ is the measure on $\left(
\mathcal{H}_{\mathbb{R}}^{^{\ast}}\left(  \infty\right)  ,\mathcal{B}\right)
$ defined by
\begin{equation}
\mathrm{P}_{\boldsymbol{\Phi}}\left(  A\right)  =\mathrm{P}_{\Psi}\left(
\widetilde{\mathcal{G}}^{-1}\left(  A\right)  \right)  \text{, for }%
A\in\mathcal{B}\text{.} \label{measure_Phi}%
\end{equation}

\begin{proposition}
\label{Prop3}The Fourier transform of $\mathrm{P}_{\boldsymbol{\Phi}}$ is
given by%
\[%
{\displaystyle\int\limits_{\mathcal{H}_{\mathbb{R}}^{^{\ast}}(\infty)}}
e^{i\left\langle f,T\right\rangle }d\mathrm{P}_{\boldsymbol{\Phi}}\left(
T\right)  =\exp\left\{
{\displaystyle\int\limits_{\mathbb{Q}_{p}^{N}}}
\Psi\left\{
{\displaystyle\int\limits_{\mathbb{Q}_{p}^{N}}}
G\left(  x-y;m,\alpha\right)  f\left(  y\right)  d^{N}y\right\}
d^{N}x\right\}  \text{,}%
\]
for $f\in\mathcal{H}_{\mathbb{R}}\left(  \infty\right)  $.
\end{proposition}

\begin{proof}
For $f\in\mathcal{H}_{\mathbb{R}}\left(  \infty\right)  $, by
(\ref{measure_Phi}) and Theorem \ref{Thm5}, we get that%
\begin{align*}%
{\displaystyle\int\limits_{\mathcal{H}_{\mathbb{R}}^{^{\ast}}(\infty)}}
e^{i\left\langle f,T\right\rangle }d\mathrm{P}_{\boldsymbol{\Phi}}\left(
T\right)   &  =%
{\displaystyle\int\limits_{\mathcal{H}_{\mathbb{R}}^{^{\ast}}(\infty)}}
e^{i\left\langle f,\widetilde{\mathcal{G}}T\right\rangle }d\mathrm{P}_{\Psi
}\left(  T\right)  =%
{\displaystyle\int\limits_{\mathcal{H}_{\mathbb{R}}^{^{\ast}}(\infty)}}
e^{i\left\langle \mathcal{G}f,T\right\rangle }d\mathrm{P}_{\Psi}\left(
T\right) \\
&  =\exp\left\{
{\displaystyle\int\limits_{\mathbb{Q}_{p}^{N}}}
\Psi\left\{
{\displaystyle\int\limits_{\mathbb{Q}_{p}^{N}}}
G\left(  x-y;m,\alpha\right)  f\left(  y\right)  d^{N}y\right\}
d^{N}x\right\}  .
\end{align*}

\end{proof}

By Proposition \ref{Prop3}, the associated coordinate process%
\[
\boldsymbol{\Phi}:\mathcal{H}_{\mathbb{R}}\left(  \infty\right)  \times\left(
\mathcal{H}_{\mathbb{R}}^{^{\ast}}\left(  \infty\right)  ,\mathcal{B}\right)
\rightarrow\mathbb{R}%
\]
given by $\boldsymbol{\Phi}\left(  f,T\right)  =\left\langle \mathcal{G}%
f,T\right\rangle $, $f\in\mathcal{H}_{\mathbb{R}}\left(  \infty\right)  $,
$T\in\left(  \mathcal{H}_{\mathbb{R}}^{^{\ast}}\left(  \infty\right)
,\mathcal{B}\right)  $, is a random field on $\left(  \mathcal{H}_{\mathbb{R}%
}^{^{\ast}}\left(  \infty\right)  ,\mathcal{B},\mathrm{P}_{\boldsymbol{\Phi}%
}\right)  $. In fact, $\boldsymbol{\Phi}$ is nothing but $\widetilde
{\mathcal{G}}\boldsymbol{\digamma}$ which is defined by%
\[
\widetilde{\mathcal{G}}\boldsymbol{\digamma}\left(  f,T\right)
=\boldsymbol{\digamma}\left(  \mathcal{G}f,T\right)  \text{, }f\in
\mathcal{H}_{\mathbb{R}}\left(  \infty\right)  \text{, }T\in\mathcal{H}%
_{\mathbb{R}}^{^{\ast}}\left(  \infty\right)  .
\]
It is useful to see $\boldsymbol{\Phi}$ as the unique solution, in law, of the
stochastic equation
\[
\left(  \boldsymbol{L}_{\alpha}+m^{2}\right)  \boldsymbol{\Phi}%
=\boldsymbol{\digamma},
\]
where $\left(  \boldsymbol{L}_{\alpha}+m^{2}\right)  \boldsymbol{\Phi}\left(
f,T\right)  :=\boldsymbol{\Phi}\left(  \left(  \boldsymbol{L}_{\alpha}%
+m^{2}\right)  f,T\right)  $, for $f\in\mathcal{H}_{\mathbb{R}}\left(
\infty\right)  $, $T\in\mathcal{H}_{\mathbb{R}}^{^{\ast}}\left(
\infty\right)  $. We note that the correctness of this last definition is a
consequence of Corollary \ref{Cor4} and Theorem \ref{Thm3} (ii).

\subsection{Symmetries}

Given a polynomial $\mathfrak{a}\left(  \xi\right)  \in\mathbb{Q}_{p}\left[
\xi_{1},\cdots,\xi_{n}\right]  $ and $\boldsymbol{g}\in GL_{N}\left(
\mathbb{Q}_{p}\right)  $, we say that $\boldsymbol{g}$ \textit{preserves}
$\mathfrak{a}$ if $\mathfrak{a}\left(  \xi\right)  =\mathfrak{a}\left(
\boldsymbol{g}\xi\right)  $, for all $\xi\in\mathbb{Q}_{p}^{N}$. By
simplicity, we use $\boldsymbol{g}x$ to mean $\left[  g_{ij}\right]  x^{T}$,
$x=\left(  x_{1},\cdots,x_{N}\right)  \in\mathbb{Q}_{p}^{N}$, where we
identify $\boldsymbol{g}$ with the matrix $\left[  g_{ij}\right]  $.

\begin{definition}
\label{def_6A}Let $\mathfrak{q}\left(  \xi\right)  $ be the elliptic quadratic
form used in the definition of the Fourier transform, see Section
\ref{Sect_Fourier_Trans}, and let $\mathfrak{l}\left(  \xi\right)  $ be the
elliptic polynomial that appears in the symbol of the operator $\boldsymbol{L}%
_{\alpha}$, see Section \ref{Sect4}. We define the homogeneous Euclidean group
of $\mathbb{Q}_{p}^{N}$ relative to $\mathfrak{q}\left(  \xi\right)  $\ and
$\mathfrak{l}\left(  \xi\right)  $, denoted as $E_{0}\left(  \mathbb{Q}%
_{p}^{N};\mathfrak{q},\mathfrak{l}\right)  :=E_{0}\left(  \mathbb{Q}_{p}%
^{N}\right)  $, as the subgroup of $GL_{N}\left(  \mathbb{Q}_{p}\right)  $
whose elements preserve $\mathfrak{q}\left(  \xi\right)  $\ and $\mathfrak{l}%
\left(  \xi\right)  $ simultaneously. We define the inhomogeneous Euclidean
group, denoted as $E\left(  \mathbb{Q}_{p}^{N};\mathfrak{q},\mathfrak{l}%
\right)  \allowbreak:=E\left(  \mathbb{Q}_{p}^{N}\right)  $, to be the group
of transformations of the form $\left(  a,\boldsymbol{g}\right)
x=a+\boldsymbol{g}x$, for $a,x\in\mathbb{Q}_{p}^{N}$, $\boldsymbol{g}\in
E_{0}\left(  \mathbb{Q}_{p}^{N}\right)  .$
\end{definition}

We notice that $E\left(  \mathbb{Q}_{p}^{N};\mathfrak{q},\mathfrak{l}\right)
$ preserves de Haar measure. Indeed, \ take $x=a+\boldsymbol{g}y$ with $g\in
E_{0}\left(  \mathbb{Q}_{p}^{N};\mathfrak{q},\mathfrak{l}\right)  \subset
O(\mathfrak{q})$, the orthogonal group of $\mathfrak{q}$, then $d^{N}%
x=d^{N}\left(  \boldsymbol{g}y\right)  =\left\vert \det\boldsymbol{g}%
\right\vert _{p}d^{N}y=d^{N}y$. In addition, it is not a straightforward
matter to decide whether or not $E_{0}\left(  \mathbb{Q}_{p}^{N}%
;\mathfrak{q},\mathfrak{l}\right)  $ is non trivial. We also notice that
$\left(  a,\boldsymbol{g}\right)  ^{-1}x=\boldsymbol{g}^{-1}\left(
x-a\right)  $.

Let $\left(  a,\boldsymbol{g}\right)  $ be a transformation in $E\left(
\mathbb{Q}_{p}^{N}\right)  $, the action of $\left(  a,\boldsymbol{g}\right)
$ on a function $f\in\mathcal{H}_{\mathbb{R}}\left(  \infty\right)  $ is
defined by%
\[
\left(  \left(  a,\boldsymbol{g}\right)  f\right)  \left(  x\right)  =f\left(
\left(  a,\boldsymbol{g}\right)  ^{-1}x\right)  \text{, for\ }x\in
\mathbb{Q}_{p}^{N},
\]
and on a functional $T\in\mathcal{H}_{\mathbb{R}}^{^{\ast}}\left(
\infty\right)  $, by%
\[
\left\langle f,\left(  a,\boldsymbol{g}\right)  T\right\rangle :=\left\langle
\left(  a,\boldsymbol{g}\right)  ^{-1}f,T\right\rangle \text{, for }%
f\in\mathcal{H}_{\mathbb{R}}\left(  \infty\right)  .
\]
The action on a random field $\boldsymbol{\Phi}$ is defined by%
\[
\left(  \left(  a,\boldsymbol{g}\right)  \boldsymbol{\Phi}\right)  \left(
f,T\right)  =\boldsymbol{\Phi}\left(  \left(  a,\boldsymbol{g}\right)
^{-1}f,T\right)  \text{, for }f\in\mathcal{H}_{\mathbb{R}}\left(
\infty\right)  \text{, }T\in\mathcal{H}_{\mathbb{R}}^{^{\ast}}\left(
\infty\right)  .
\]

\begin{definition}
\label{def_6}By Euclidean invariance of the random field $\boldsymbol{\Phi}$
we mean that the laws of $\boldsymbol{\Phi}$ and $\left(  a,\boldsymbol{g}%
\right)  \boldsymbol{\Phi}$ are the same for each\ $\left(  a,\boldsymbol{g}%
\right)  \in E\left(  \mathbb{Q}_{p}^{N}\right)  $, i.e. the probability
distributions of $\left\{  \boldsymbol{\Phi}\left(  f,\cdot\right)
:f\in\mathcal{H}_{\mathbb{R}}\left(  \infty\right)  \right\}  $ and $\left\{
\left(  \left(  a,\boldsymbol{g}\right)  \boldsymbol{\Phi}\right)  \left(
f,\cdot\right)  :f\in\mathcal{H}_{\mathbb{R}}\left(  \infty\right)  \right\}
$ coincide for each $\left(  a,\boldsymbol{g}\right)  \in E\left(
\mathbb{Q}_{p}^{N}\right)  $.
\end{definition}

We say that $\mathcal{G}$ is $\left(  a,\boldsymbol{g}\right)  $-invariant for
some $\left(  a,\boldsymbol{g}\right)  \in E\left(  \mathbb{Q}_{p}^{N}\right)
$, if $\left(  a,\boldsymbol{g}\right)  \mathcal{G}=\mathcal{G}\left(
a,\boldsymbol{g}\right)  $. If $\mathcal{G}$ is invariant under all $\left(
a,\boldsymbol{g}\right)  \in E\left(  \mathbb{Q}_{p}^{N}\right)  $, we say
that $\mathcal{G}$ is \textit{Euclidean invariant.}

\begin{remark}
\label{nota_Fourier}Let $f\in\mathcal{D}_{\mathbb{C}}\left(  \mathbb{Q}%
_{p}^{N}\right)  $ and let $\left(  a,\boldsymbol{g}\right)  \in E\left(
\mathbb{Q}_{p}^{N}\right)  $. Then
\begin{equation}
\mathcal{F}_{x\rightarrow\xi}\left[  f\left(  \left(  a,\boldsymbol{g}\right)
^{-1}x\right)  \right]  =\chi_{p}\left(  \mathfrak{B}\left(  a,\xi\right)
\right)  \mathcal{F}_{x\rightarrow\xi}\left[  f\right]  \left(  \boldsymbol{g}%
^{-1}\xi\right)  . \label{formula_F}%
\end{equation}

Indeed, by taking $\boldsymbol{g}^{-1}\left(  x-a\right)  x=y$,%
\[%
{\displaystyle\int\limits_{\mathbb{Q}_{p}^{N}}}
f\left(  \boldsymbol{g}^{-1}\left(  x-a\right)  x\right)  \chi_{p}\left(
\mathfrak{B}\left(  x,\xi\right)  \right)  d\mu(x)=\chi_{p}\left(
\mathfrak{B}\left(  a,\xi\right)  \right)
{\displaystyle\int\limits_{\mathbb{Q}_{p}^{N}}}
f\left(  y\right)  \chi_{p}\left(  \mathfrak{B}\left(  \boldsymbol{g}%
y,\xi\right)  \right)  d\mu(y),
\]
where $d\mu(y)=C(\mathfrak{q})d^{N}y$. The formula follows from%
\begin{align*}
\mathfrak{B}\left(  \boldsymbol{g}y,\xi\right)   &  =\frac{1}{2}\left\{
\mathfrak{q}\left(  \boldsymbol{g}y+\xi\right)  -\mathfrak{q}\left(
\boldsymbol{g}y\right)  -\mathfrak{q}\left(  \xi\right)  \right\} \\
&  =\frac{1}{2}\left\{  \mathfrak{q}\left(  y+\boldsymbol{g}^{-1}\xi\right)
-\mathfrak{q}\left(  y\right)  -\mathfrak{q}\left(  \boldsymbol{g}^{-1}%
\xi\right)  \right\}  =\mathfrak{B}\left(  y,\boldsymbol{g}^{-1}\xi\right)  .
\end{align*}

By the density of $\mathcal{D}_{\mathbb{C}}\left(  \mathbb{Q}_{p}^{N}\right)
$ in $\mathcal{D}_{\mathbb{C}}^{^{\ast}}\left(  \mathbb{Q}_{p}^{N}\right)  $,
formula (\ref{formula_F}) holds in $\mathcal{D}_{\mathbb{C}}^{^{\ast}}\left(
\mathbb{Q}_{p}^{N}\right)  $.
\end{remark}

\begin{lemma}
\label{lemma7}$\mathcal{G}$ is \textit{Euclidean invariant.}
\end{lemma}

\begin{proof}
We first notice that the mapping $f\left(  x\right)  \rightarrow f\left(
\left(  a,\boldsymbol{g}\right)  ^{-1}x\right)  $ is continuos from
$\mathcal{H}_{\mathbb{R}}\left(  \infty\right)  $ into $\mathcal{H}%
_{\mathbb{R}}\left(  \infty\right)  $, for any $\left(  a,\boldsymbol{g}%
\right)  \in E\left(  \mathbb{Q}_{p}^{N}\right)  $, then $\mathcal{G}\left(
a,\boldsymbol{g}\right)  $, $\left(  a,\boldsymbol{g}\right)  \mathcal{G}%
:\mathcal{H}_{\mathbb{R}}\left(  \infty\right)  \rightarrow\mathcal{H}%
_{\mathbb{R}}\left(  \infty\right)  $ are continuous, cf. Corollary
\ref{Cor4}, and in order to show that%
\[
\left(  \left(  a,\boldsymbol{g}\right)  \mathcal{G}\right)  \left(  f\right)
=\left(  \mathcal{G}\left(  a,\boldsymbol{g}\right)  \right)  \left(
f\right)  ,
\]
it is sufficient to take $f\in\mathcal{D}_{\mathbb{R}}\left(  \mathbb{Q}%
_{p}^{N}\right)  $. Now, since%
\[
\left(  \left(  a,\boldsymbol{g}\right)  \mathcal{G}\right)  \left(  f\right)
\left(  x\right)  =\left(  a,\boldsymbol{g}\right)  \left(  G\left(
x;m,\alpha\right)  \ast f\left(  x\right)  \right)  =\left(  G\ast f\right)
\left(  \left(  a,\boldsymbol{g}\right)  ^{-1}x\right)  ,
\]
and
\[
\left(  \mathcal{G}\left(  a,\boldsymbol{g}\right)  \right)  \left(  f\right)
\left(  x\right)  =\mathcal{G}\left(  f\left(  \left(  a,\boldsymbol{g}%
\right)  ^{-1}x\right)  \right)  =G\left(  x;m,\alpha\right)  \ast f\left(
\left(  a,\boldsymbol{g}\right)  ^{-1}x\right)  ,
\]
we should show that $\left(  G\ast f\right)  \left(  \left(  a,\boldsymbol{g}%
\right)  ^{-1}x\right)  =G\left(  x;m,\alpha\right)  \ast f\left(  \left(
a,\boldsymbol{g}\right)  ^{-1}x\right)  $. We establish this formula in
$\mathcal{D}_{\mathbb{C}}^{^{\ast}}\left(  \mathbb{Q}_{p}^{N}\right)  $ by
using the Fourier transform. Indeed,%
\begin{multline*}
\mathcal{F}\left[  \left(  G\ast f\right)  \left(  \left(  a,\boldsymbol{g}%
\right)  ^{-1}x\right)  \right]  =\chi_{p}\left(  \mathfrak{B}\left(
a,\xi\right)  \right)  \mathcal{F}\left[  \left(  G\ast f\right)  \right]
\left(  \boldsymbol{g}^{-1}\xi\right)  =\\
\chi_{p}\left(  \mathfrak{B}\left(  a,\xi\right)  \right)  \frac
{\mathcal{F}\left[  f\right]  \left(  \boldsymbol{g}^{-1}\xi\right)
}{\left\vert \mathfrak{l}\left(  \boldsymbol{g}^{-1}\xi\right)  \right\vert
_{p}^{\alpha}+m^{2}}=\chi_{p}\left(  \mathfrak{B}\left(  a,\xi\right)
\right)  \frac{\mathcal{F}\left[  f\right]  \left(  \boldsymbol{g}^{-1}%
\xi\right)  }{\left\vert \mathfrak{l}\left(  \xi\right)  \right\vert
_{p}^{\alpha}+m^{2}},
\end{multline*}
and
\begin{multline*}
\mathcal{F}\left[  G\left(  x;m,\alpha\right)  \ast f\left(  \boldsymbol{g}%
^{-1}x\right)  \right]  =\frac{1}{\left\vert \mathfrak{l}\left(  \xi\right)
\right\vert _{p}^{\alpha}+m^{2}}\mathcal{F}_{x\rightarrow\xi}\left[  f\left(
\boldsymbol{g}^{-1}x\right)  \right] \\
=\chi_{p}\left(  \mathfrak{B}\left(  a,\xi\right)  \right)  \frac
{\mathcal{F}\left[  f\right]  \left(  \boldsymbol{g}^{-1}\xi\right)
}{\left\vert \mathfrak{l}\left(  \xi\right)  \right\vert _{p}^{\alpha}+m^{2}},
\end{multline*}
cf. Remark \ref{nota_Fourier}.
\end{proof}

\begin{proposition}
\label{Prop4}The random field $\boldsymbol{\Phi}=$\ $\widetilde{\mathcal{G}%
}\boldsymbol{\digamma}$ is Euclidean invariant.
\end{proposition}

\begin{proof}
By Bochner-Minlos Theorem, it is sufficient to show that%
\begin{equation}
C_{\boldsymbol{\Phi}}\left(  f\right)  =C_{\left(  a,\boldsymbol{g}\right)
\boldsymbol{\Phi}}\left(  f\right)  \text{, for }f\in\mathcal{H}_{\mathbb{R}%
}\left(  \infty\right)  \text{ and for every }\left(  a,\boldsymbol{g}\right)
\in E\left(  \mathbb{Q}_{p}^{N}\right)  . \label{C_phi_Invariant}%
\end{equation}
Indeed,%
\begin{multline*}
C_{\left(  a,\boldsymbol{g}\right)  \boldsymbol{\Phi}}\left(  f\right)  =%
{\displaystyle\int\limits_{\mathcal{H}_{\mathbb{R}}^{^{\ast}}(\infty)}}
e^{i\left(  \left(  a,\boldsymbol{g}\right)  \boldsymbol{\Phi}\right)  \left(
f,T\right)  }d\mathrm{P}_{\boldsymbol{\Phi}}\left(  T\right)  =%
{\displaystyle\int\limits_{\mathcal{H}_{\mathbb{R}}^{^{\ast}}(\infty)}}
e^{i\left\langle \left(  a,\boldsymbol{g}\right)  ^{-1}f,\ \widetilde
{\mathcal{G}}T\right\rangle }d\mathrm{P}_{\Psi}\left(  T\right) \\
=%
{\displaystyle\int\limits_{\mathcal{H}_{\mathbb{R}}^{^{\ast}}(\infty)}}
e^{i\left\langle \mathcal{G}\left(  \left(  a,\boldsymbol{g}\right)
^{-1}f\right)  ,\ T\right\rangle }d\mathrm{P}_{\Psi}\left(  T\right)
=\exp\left\{
{\displaystyle\int\limits_{\mathbb{Q}_{p}^{N}}}
\Psi\left(  \mathcal{G}\left(  \left(  a,\boldsymbol{g}\right)  ^{-1}f\right)
\left(  x\right)  \right)  d^{N}x\right\}  ,
\end{multline*}
cf. Theorem \ref{Thm5} and Corollary \ref{Cor4}. Since $\mathcal{G}\left(
\left(  a,\boldsymbol{g}\right)  ^{-1}f\right)  \left(  x\right)  =\left(
a,\boldsymbol{g}\right)  ^{-1}\left(  \mathcal{G}\left(  f\right)  \right)
\left(  x\right)  $ (cf. Lemma \ref{lemma7}), and $\left(  a,\boldsymbol{g}%
\right)  $ preserves $d^{N}x$, we have%
\begin{align*}
C_{\left(  a,\boldsymbol{g}\right)  \boldsymbol{\Phi}}\left(  f\right)   &
=\exp\left\{
{\displaystyle\int\limits_{\mathbb{Q}_{p}^{N}}}
\Psi\left(  \mathcal{G}\left(  f\right)  \left(  \left(  a,\boldsymbol{g}%
\right)  x\right)  \right)  d^{N}x\right\}  =\exp\left\{
{\displaystyle\int\limits_{\mathbb{Q}_{p}^{N}}}
\Psi\left(  \mathcal{G}\left(  f\right)  \left(  x\right)  \right)
d^{N}x\right\} \\
&  =C_{\boldsymbol{\Phi}}\left(  f\right)  ,
\end{align*}
cf. Proposition \ref{Prop3}.
\end{proof}

\subsection{Some additional remarks and examples}

In the Archimedean case the symmetric bilinear form used in the definition of
the Fourier transform has associated a quadratic form which is exactly the
symbol of the Laplacian when it is considered as a pseudodifferential
operator. This approach cannot be carried out in the $p$-adic setting. Indeed,
the quadratic form $\xi_{1}^{2}+\cdots+\xi_{N}^{2}$ associated to the bilinear
form $\sum_{i}\xi_{i}x_{i}$ does not give rise to an elliptic operator if
$N\geq5$. This is the reason why in the $p$-adic setting we need two different
polynomials $\mathfrak{q}\left(  \xi\right)  $ and $\mathfrak{l}\left(
\xi\right)  $. In order to have a `non-trivial' group of symmetries, i.e.
$E_{0}\left(  \mathbb{Q}_{p}^{N},\mathfrak{q},\mathfrak{l}\right)  \neq1$, the
polynomials $\mathfrak{q}\left(  \xi\right)  $ and $\mathfrak{l}\left(
\xi\right)  $ should be related `nicely'. To illustrate this idea we give two examples.

\begin{example}
\label{example1}In the case $N=4$ there is a unique elliptic quadratic form,
up to linear equivalence, which is $\mathfrak{l}_{4}\left(  \xi\right)
=\xi_{1}^{2}-s\xi_{2}^{2}-p\xi_{3}^{2}+s\xi_{4}^{2}$, where $s\in
\mathbb{Z\smallsetminus}\left\{  0\right\}  $ is a quadratic non-residue, i.e.
$\left(  \frac{s}{p}\right)  =-1$. We take $\mathfrak{q}_{4}\left(
\xi\right)  =\mathfrak{l}_{4}\left(  \xi\right)  $, i.e. $\mathfrak{B}%
_{4}\left(  x,\xi\right)  =\xi_{1}x_{1}-s\xi_{2}x_{2}-p\xi_{3}x_{3}+s\xi
_{4}x_{4}$. In this case, $E_{0}\left(  \mathbb{Q}_{p}^{4},\mathfrak{l}%
_{4},\mathfrak{l}_{4}\right)  $ equals%
\[
O\left(  \mathfrak{l}_{4}\right)  =\left\{  \boldsymbol{g}\in GL_{4}\left(
\mathbb{Q}_{p}\right)  :\boldsymbol{g}^{T}\text{{\small $\left[
\begin{array}
[c]{cccc}%
1 & 0 & 0 & 0\\
0 & -s & 0 & 0\\
0 & 0 & -p & 0\\
0 & 0 & 0 & s
\end{array}
\right]  $}}\boldsymbol{g=}{\small \left[
\begin{array}
[c]{cccc}%
1 & 0 & 0 & 0\\
0 & -s & 0 & 0\\
0 & 0 & -p & 0\\
0 & 0 & 0 & s
\end{array}
\right]  }\right\}  ,
\]
the orthogonal group of $\mathfrak{l}_{4}$.
\end{example}

\begin{example}
\label{example2}Take $N=5$, $\mathfrak{B}_{5}\left(  x,\xi\right)  =\xi
_{1}x_{1}-s\xi_{2}x_{2}-p\xi_{3}x_{3}+s\xi_{4}x_{4}+\xi_{5}x_{5}$, and
$\mathfrak{l}_{5}\left(  \xi\right)  =\left(  \xi_{1}^{2}-s\xi_{2}^{2}%
-p\xi_{3}^{2}+s\xi_{4}^{2}\right)  ^{2}-\tau\xi_{5}^{4}$, where $s$ is as in
Example \ref{example2}\ and \ $\tau\notin\left[  \mathbb{Q}_{p}^{\times
}\right]  ^{2}$, where $\left[  \mathbb{Q}_{p}^{\times}\right]  ^{2}$\ denotes
the group of squares of $\mathbb{Q}_{p}^{\times}$. We notice that%
\[
\left\{  \left[
\begin{array}
[c]{cc}%
\boldsymbol{g} & \boldsymbol{0}\\
\boldsymbol{0} & \left[  1\right]  _{1\times1}%
\end{array}
\right]  _{5\times5}\in GL_{5}\left(  \mathbb{Q}_{p}\right)  :\boldsymbol{g}%
\in O\left(  \mathfrak{l}_{4}\right)  \right\}  \subseteq E_{0}\left(
\mathbb{Q}_{p}^{5},\mathfrak{q}_{5},\mathfrak{l}_{5}\right)  .
\]

\end{example}

\begin{example}
\label{example3}Take $\mathfrak{B}\left(  x,\xi\right)  =\sum_{i=1}^{N}\xi
_{i}x_{i}$. Then the theory developed so far can be applied to
pseudodifferential operators of type%
\[
g\rightarrow\mathcal{F}^{-1}\left(  \left\Vert \xi\right\Vert _{p}^{\alpha
}\mathcal{F}g\right)  \text{, or }g\rightarrow\mathcal{F}^{-1}\left(  \left(
\sum_{i=1}^{N}\left\vert \xi_{i}\right\vert _{p}\right)  ^{\alpha}%
\mathcal{F}g\right)  \text{, with }\alpha>0.
\]
We notice that the group of permutations of the variables $\xi_{1},\cdots
,\xi_{N}$ preserve $\left\Vert \xi\right\Vert _{p}^{\alpha}$ and $\left(
\sum_{i=1}^{N}\left\vert \xi_{i}\right\vert _{p}\right)  ^{\alpha}$.
\end{example}

\section{\label{Sect7}Schwinger Functions}

\begin{remark}
\textbf{Gâteaux derivatives. }Let $y\in\mathcal{H}_{\mathbb{R}}(\infty)$ be
fixed. Assume that $\boldsymbol{F}$:$\mathcal{H}_{\mathbb{R}}(\infty
)\rightarrow\mathbb{R}$. For $x\in\mathcal{H}_{\mathbb{R}}(\infty)$ consider
$\lambda\rightarrow\boldsymbol{F}\left(  x+\lambda y\right)  $ on $\mathbb{R}%
$. If this function is differentiable at $\lambda=0$, we say that
$\boldsymbol{F}$ is Gâteaux differentiable at $x$ in the direction $y$ and
denote $D_{y}F\left(  x\right)  :=\frac{d}{d\lambda}\boldsymbol{F}\left(
x+\lambda y\right)  \mid_{\lambda=0}$. $D_{y}$ acts linearly and admits the
usual chain and product rules. We define \textit{the partial derivative of the
functional }$\boldsymbol{F}$\textit{ with respect to }$g\in\mathcal{H}%
_{\mathbb{R}}(\infty)$ as
\begin{equation}
\frac{\partial}{\partial g}\boldsymbol{F}=D_{g}\boldsymbol{F}.
\label{derivative}%
\end{equation}

\end{remark}

\begin{definition}
Let $g_{1},\cdots,g_{m}\in\mathcal{H}_{\mathbb{R}}(\infty)$. We define
$M_{m}^{\digamma}$, the $m-th$ moment of of the random field $\digamma$, by%
\[
M_{m}^{\digamma}\left(  g_{1}\otimes\cdots\otimes g_{m}\right)  =%
{\displaystyle\int\limits_{\mathcal{H}_{\mathbb{R}}^{^{\ast}}(\infty)}}
\left\langle g_{1},T\right\rangle \cdots\left\langle g_{m},T\right\rangle
d\mathrm{P}_{\Psi}\left(  T\right)  \text{, }m\in\mathbb{N}\smallsetminus
\left\{  0\right\}  ,\text{ }%
\]
and $M_{0}^{\digamma}:=1$.
\end{definition}

\begin{lemma}
\label{lemma8}(i) Let%
\[%
\begin{array}
[c]{cccc}%
C_{\digamma}^{T}: & \mathcal{H}_{\mathbb{R}}(\infty) & \rightarrow &
\mathbb{R}\\
& g & \rightarrow & \int_{\mathbb{Q}_{p}^{N}}\Psi\left(  g\left(  x\right)
\right)  d^{N}x.
\end{array}
\]
Then partial derivatives of all orders of $C_{\digamma}^{T}$ exist everywhere
on $\mathcal{H}_{\mathbb{R}}(\infty)$. For $g_{1},\cdots,g_{m}\in
\mathcal{H}_{\mathbb{R}}(\infty)$, we have
\[
\frac{1}{i^{m}}\frac{\partial^{m}}{\partial g_{m}\cdots\partial g_{1}%
}C_{\digamma}^{T}\mid_{0}=c_{m}%
{\displaystyle\int\limits_{\mathbb{Q}_{p}^{N}}}
g_{1}\cdots g_{m}d^{N}x,
\]
where%
\begin{equation}%
\begin{array}
[c]{lll}%
c_{1}:=a+%
{\displaystyle\int\limits_{\mathbb{R}\smallsetminus\left\{  0\right\}  }}
\frac{s^{3}}{1+s^{2}}dM\left(  s\right)  , &  & c_{2}:=\sigma^{2}+%
{\displaystyle\int\limits_{\mathbb{R}\smallsetminus\left\{  0\right\}  }}
s^{2}dM\left(  s\right)  ,\\
&  & \\
c_{m}:=%
{\displaystyle\int\limits_{\mathbb{R}\smallsetminus\left\{  0\right\}  }}
s^{m}dM\left(  s\right)  , & \text{for }m\geq3. &
\end{array}
\label{C_s}%
\end{equation}

(ii) Take $C_{\digamma}=\exp C_{\digamma}^{T}$. Then $C_{\digamma}$ has
partial derivatives of any order, and the moments of $\digamma$ satisfy%
\begin{equation}
M_{m}^{\digamma}\left(  g_{1}\otimes\cdots\otimes g_{m}\right)  =\frac
{1}{i^{m}}\frac{\partial^{m}}{\partial g_{m}\cdots\partial g_{1}}C_{\digamma
}\mid_{0}. \label{Moment_1}%
\end{equation}

\end{lemma}

\begin{proof}
(i) See the proof of Lemma 3.1 in \cite{Albeverio-et-al-3}. (ii) The formula
follows from Theorem \ref{Thm5} and (i), by
\[
\frac{\partial^{m}}{\partial g_{m}\cdots\partial g_{1}}%
{\textstyle\int\limits_{\mathcal{H}_{\mathbb{R}}^{^{\ast}}\left(
\infty\right)  }}
e^{i\left\langle g,T\right\rangle }\mathrm{dP}_{\Psi}\left(  T\right)
=i^{m}M_{m}^{\digamma}\left(  g_{1}\otimes\cdots\otimes g_{m}\right)  .
\]
This last formula follows from the Dominated Convergence Theorem by using that
$\left\vert e^{i\left\langle g,T\right\rangle }\right\vert =1$ and that
$\mathrm{P}_{\Psi}$ is a probability measure.
\end{proof}

\begin{lemma}
[{\cite[Corollary 3.5]{Albeverio-et-al-3}}]\label{Corollary 3.5}Let $V$ be a
vector space. Let $g:V\rightarrow\mathbb{C}$ be \ infinitely often
differentiable on $V$, in the sense of (\ref{derivative}), such that $g(0)=0$.
Let $f$ the exponential function. Then $f^{\left(  k\right)  }\circ g\left(
0\right)  =1$ for all $k\in\mathbb{N}$. Let $P^{\left(  m\right)  }$ be the
collection of all partitions $I$ of\ $\left\{  1,\cdots,m\right\}  $\ into
disjoint sets. It follows that for $v_{1},\cdots,v_{m}\in V$, we get%
\[
\frac{\partial^{m}}{\partial v_{m}\cdots\partial v_{1}}\exp g\mid_{0}=%
{\displaystyle\sum\limits_{I\in P^{\left(  m\right)  }}}
\text{ }%
{\displaystyle\prod\limits_{\left\{  j_{1},\cdots,j_{l}\right\}  \in I}}
\frac{\partial^{l}}{\partial v_{j_{1}}\cdots\partial v_{j_{l}}}g\mid_{0}.
\]

\end{lemma}

\begin{proposition}
\label{Prop5}Set $g_{1},\cdots,g_{m}\in\mathcal{H}_{\mathbb{R}}(\infty)$. Then%
\[
M_{m}^{\digamma}\left(  g_{1}\otimes\cdots\otimes g_{m}\right)  =%
{\displaystyle\sum\limits_{I\in P^{\left(  m\right)  }}}
\text{ }%
{\displaystyle\prod\limits_{\left\{  j_{1},\cdots,j_{l}\right\}  \in I}}
c_{l}%
{\displaystyle\int\limits_{\mathbb{Q}_{p}^{N}}}
g_{j_{1}}\cdots g_{j_{l}}d^{N}x.
\]

\end{proposition}

\begin{proof}
The formula follows from (\ref{Moment_1}), $C_{\digamma}=\exp C_{\digamma}%
^{T}$, Lemma \ref{lemma8}\ (i)\ and Lemma \ref{Corollary 3.5}.
\end{proof}

\begin{definition}
\label{def6}Set $g_{1},\cdots,g_{m}\in\mathcal{H}_{\mathbb{R}}(\infty)$. We
define the $m-th$ Schwinger function $S_{m}$ of $\boldsymbol{\Phi}$ as the
$m-th$ moment of $\boldsymbol{\Phi}$, i.e.%
\begin{equation}
S_{m}\left(  g_{1}\otimes\cdots\otimes g_{m}\right)  =%
{\displaystyle\int\limits_{\mathcal{H}_{\mathbb{R}}^{^{\ast}}(\infty)}}
\left\langle g_{1},T\right\rangle \cdots\left\langle g_{m},T\right\rangle
d\mathrm{P}_{\boldsymbol{\Phi}}\left(  T\right)  \text{, }m\in\mathbb{N}%
\smallsetminus\left\{  0\right\}  , \label{Def_S_m}%
\end{equation}
with $S_{0}:=1$.
\end{definition}

\begin{theorem}
\label{Thm6}The Schwinger functions $S_{m}$ defined above are symmetric and
Euclidean invariant functionals in $\mathcal{H}_{\mathbb{R}}^{^{\ast}}\left(
\mathbb{Q}_{p}^{Nm};\infty\right)  $ for $m\geq1$. Furthermore for
$g_{1},\cdots,g_{m}\in\mathcal{H}_{\mathbb{R}}(\mathbb{Q}_{p}^{N},\infty)$ we
have%
\begin{equation}
S_{m}\left(  g_{1}\otimes\cdots\otimes g_{m}\right)  =%
{\displaystyle\sum\limits_{I\in P^{\left(  m\right)  }}}
\text{ }%
{\displaystyle\prod\limits_{\left\{  j_{1},\cdots,j_{l}\right\}  \in I}}
\text{ }c_{l}%
{\displaystyle\int\limits_{\mathbb{Q}_{p}^{N}}}
\left(
{\displaystyle\prod\limits_{k=1}^{l}}
G\left(  x;m,\alpha\right)  \ast g_{j_{k}}\left(  x\right)  \right)  d^{N}x,
\label{formula_S_m}%
\end{equation}
where the constants $c_{l}$ were defined in (\ref{C_s}).
\end{theorem}

\begin{proof}
The symmetry follows from Definition \ref{def6}. To establish the Euclidean
invariance, we proceed as follows. By (\ref{C_phi_Invariant})-(\ref{Moment_1}%
),%
\[
\frac{1}{i^{m}}\frac{\partial^{m}}{\partial f_{m}\cdots\partial f_{1}%
}C_{\boldsymbol{\Phi}}\left(  f\right)  \mid_{0}=\frac{1}{i^{m}}\frac
{\partial^{m}}{\partial f_{m}\cdots\partial f_{1}}C_{\left(  a,\boldsymbol{g}%
\right)  \boldsymbol{\Phi}}\left(  f\right)  \mid_{0},
\]
for any $f,f_{1},\cdots,f_{m}\in\mathcal{H}_{\mathbb{R}}(\mathbb{Q}_{p}%
^{N},\infty)$ and $\left(  a,\boldsymbol{g}\right)  \in E\left(
\mathbb{Q}_{p}^{N}\right)  $. Therefore%
\[%
{\displaystyle\int\limits_{\mathcal{H}_{\mathbb{R}}^{^{\ast}}(\infty)}}
\left\langle g_{1},T\right\rangle \cdots\left\langle g_{m},T\right\rangle
d\mathrm{P}_{\boldsymbol{\Phi}}\left(  T\right)  =%
{\displaystyle\int\limits_{\mathcal{H}_{\mathbb{R}}^{^{\ast}}(\infty)}}
\left\langle g_{1},T\right\rangle \cdots\left\langle g_{m},T\right\rangle
d\mathrm{P}_{\left(  a,\boldsymbol{g}\right)  \boldsymbol{\Phi}}\left(
T\right)  \text{, }%
\]
$m\in\mathbb{N}\smallsetminus\left\{  0\right\}  $ and any $\left(
a,\boldsymbol{g}\right)  \in E\left(  \mathbb{Q}_{p}^{N}\right)  $.

We now establish (\ref{formula_S_m}). By using $\mathrm{P}_{\boldsymbol{\Phi}%
}=\mathrm{P}_{\Psi}\circ\left(  \widetilde{\mathcal{G}}\right)  ^{-1}$, the
right hand side of (\ref{Def_S_m}) is equal to%
\[%
{\displaystyle\int\limits_{\mathcal{H}_{\mathbb{R}}^{^{\ast}}(\infty)}}
\left\langle g_{1},\widetilde{\mathcal{G}}T\right\rangle \cdots\left\langle
g_{m},\widetilde{\mathcal{G}}T\right\rangle d\mathrm{P}_{\Psi}\left(
T\right)  =%
{\displaystyle\int\limits_{\mathcal{H}_{\mathbb{R}}^{^{\ast}}(\infty)}}
\left\langle G\ast g_{1},T\right\rangle \cdots\left\langle G\ast
g_{m},T\right\rangle d\mathrm{P}_{\Psi}\left(  T\right)  ,
\]
with $G=G(x;m,\alpha)$. Now, the formula follows from Proposition \ref{Prop5}.

We set $S_{l}^{T}$ for the mapping
\[
g_{1}\otimes\cdots\otimes g_{l}\rightarrow c_{l}%
{\displaystyle\int\limits_{\mathbb{Q}_{p}^{N}}}
\left(
{\displaystyle\prod\limits_{i=1}^{l}}
G\ast g_{i}\right)  d^{N}x,
\]
where $l\in\mathbb{N}\smallsetminus\left\{  0\right\}  $. We show that
$S_{l}^{T}$ is a $l$-linear form which is continuous in each of its arguments,
i.e. separately continuous. We first show the case $l=2$. By Corollary
\ref{Cor4}, $G\ast g_{1}$, $G\ast g_{2}\in L_{\mathbb{R}}^{2}$, then%
\begin{align*}
B(g_{1},g_{2})  &  :=%
{\displaystyle\int\limits_{\mathbb{Q}_{p}^{N}}}
\left(  G\ast g_{1}\right)  \left(  G\ast g_{2}\right)  d^{N}x=\\
&
{\displaystyle\int\limits_{\mathbb{Q}_{p}^{N}}}
\left(  \widehat{G}\left(  \xi\right)  \widehat{g}_{1}\left(  \xi\right)
\right)  \left(  \overline{\widehat{G}\left(  \xi\right)  \widehat{g}%
_{2}\left(  \xi\right)  }\right)  d^{N}\xi.
\end{align*}
We now pick $k\in\mathbb{N}\smallsetminus\left\{  0\right\}  $, use
$\left\vert \widehat{G}\left(  \xi\right)  \right\vert \leq\frac{1}{m^{2}}$
and the Cauchy-Schwartz inequality, to get%
\begin{align*}
\left\vert B(g_{1},g_{2})\right\vert  &  \leq\frac{1}{m^{4}}%
{\displaystyle\int\limits_{\mathbb{Q}_{p}^{N}}}
\left(  \frac{\left\vert \widehat{g}_{1}\left(  \xi\right)  \right\vert
}{\left[  \max(1,\left\Vert \xi\right\Vert _{p})\right]  ^{\alpha k}}\right)
\left(  \left[  \max(1,\left\Vert \xi\right\Vert _{p})\right]  ^{\alpha
k}\left\vert \widehat{g}_{2}\left(  \xi\right)  \right\vert \right)  d^{N}%
\xi\\
&  \leq\frac{1}{m^{4}}\left\Vert g_{1}\right\Vert _{0}\left\Vert
g_{2}\right\Vert _{k},
\end{align*}
i.e. $B(g_{1},\cdot)$ is continuous. We now established the case $l\geq3$ by
using induction on $l$. By using that $\widehat{g}_{i}\in L_{\mathbb{R}}%
^{1}\cap L_{\mathbb{R}}^{2}$, for $i=1,\cdots,l-1$, cf. Corollary \ref{Cor4}
and Remark \ref{nota2} (ii), and that $\left\vert \widehat{G}\widehat{g}%
_{i}\right\vert \leq\frac{1}{m^{2}}\left\vert \widehat{g}_{i}\right\vert $,
for $i=1,\cdots,l-1$, we have $\widehat{G}\widehat{g}_{i}\in L_{\mathbb{R}%
}^{1}\cap L_{\mathbb{R}}^{2}$, for $i=1,\cdots,l-1$, now by applying the Young
inequality repeatedly,
\[
\widehat{H}(\xi):=\widehat{G}\widehat{g}_{1}\left(  \xi\right)  \ast\cdots
\ast\widehat{G}\widehat{g}_{l-1}\left(  \xi\right)  \in L_{\mathbb{R}}^{2}.
\]
Then $H(\xi)=%
{\textstyle\prod\nolimits_{i=1}^{l-1}}
G\ast g_{i}\in L_{\mathbb{R}}^{2}$. We now apply the argument in the case
$l=2$ to $\
{\textstyle\prod\nolimits_{i=1}^{l-1}}
G\ast g_{i}\in L_{\mathbb{R}}^{2}$ and $G\ast g_{l}$ to show that
$B(g_{1},\cdots,g_{l-1},\cdot)$ is continuous for fixed $g_{1},\cdots
,g_{l-1}\in\mathcal{H}_{\mathbb{R}}(\mathbb{Q}_{p}^{N};\infty)$. Hence
$S_{l}^{T}$ is a $l$-linear form which is separately continuous on
$\mathcal{H}_{\mathbb{R}}(\infty)\otimes_{\text{alg}}\cdots\otimes
_{\text{alg}}\mathcal{H}_{\mathbb{R}}(\mathbb{Q}_{p}^{N};\infty)$, $l$-times.
Since $\mathcal{H}_{\mathbb{R}}(\mathbb{Q}_{p}^{N};\infty)$ is a Fréchet
space, $S_{l}^{T}$ is $l$-linear continuous, see e.g. Proposition 1.3.11 in
\cite{Obata}. We now apply the Kernel Theorem inductively on $l$, see Remark
\ref{Kernel_theorem}, and Lemma \ref{lemma_tensor_product}, to get%
\begin{align*}
S_{l}^{T}  &  \in\mathcal{B}\left(  \mathcal{H}_{\mathbb{R}}(\mathbb{Q}%
_{p}^{(l-1)N};\infty),\mathcal{H}_{\mathbb{R}}(\mathbb{Q}_{p}^{N}%
;\infty)\right)  \cong\left(  \mathcal{H}_{\mathbb{R}}(\mathbb{Q}_{p}%
^{(l-1)N};\infty)\otimes_{\pi}\mathcal{H}_{\mathbb{R}}(\mathbb{Q}_{p}%
^{N};\infty)\right)  ^{\ast}\\
&  \cong\mathcal{H}_{\mathbb{R}}^{^{\ast}}(\mathbb{Q}_{p}^{lN};\infty).
\end{align*}

We now consider $%
{\textstyle\prod\nolimits_{\left\{  j_{1},\cdots,j_{l}\right\}  \in I}}
S_{l}^{T}$. By induction on the cardinality of $I$, and by using the Kernel
Theorem and Lemma \ref{lemma_tensor_product}, we show that $%
{\textstyle\prod\nolimits_{\left\{  j_{1},\cdots,j_{l}\right\}  \in I}}
S_{l}^{T}\in\mathcal{H}_{\mathbb{R}}^{^{\ast}}(\mathbb{Q}_{p}^{mN};\infty)$.
Therefore $S_{m}\in\mathcal{H}_{\mathbb{R}}^{^{\ast}}(\mathbb{Q}_{p}%
^{mN};\infty)$.
\end{proof}

\subsection{The $p$-adic Brownian sheet on $\mathbb{Q}_{p}^{N}$}

As an application of the results developed in this section, we present a
construction of the Wiener process with time variable in $\mathbb{Q}_{p}^{N}$.
In this section we take $\Psi$ with $a=0$ and $M=0$ in (\ref{Psi}). Thus, the
generalized white noise $\digamma$ in Definition \ref{def2}\ is Gaussian with
mean zero. Given $t=(t_{1},\cdots,t_{N})\in\mathbb{Q}_{p}^{N}$, we set for
$x\in\mathbb{Q}_{p}^{N}$,
\[
1_{\left[  0,t\right]  ^{N}}(x):=\left\{
\begin{array}
[c]{ll}%
1 & \text{if }\left\Vert x\right\Vert _{p}\leq\left\Vert t\right\Vert _{p}\\
0 & \text{otherwise.}%
\end{array}
\right.
\]
We also set $W\left(  t\right)  :=\left\{  W\left(  t,\cdot\right)  \right\}
_{t\in\mathbb{Q}_{p}^{N}}=\left\{  \digamma\left(  1_{\left[  0,t\right]
^{N}},\cdot\right)  \right\}  _{t\in\mathbb{Q}_{p}^{N}}$. We call the process
$W\left(  t\right)  $ with values in $\mathbb{R}$, the\textit{ }%
$p$\textit{-adic Brownian sheet on }$\mathbb{Q}_{p}^{N}$. The following result
follows from Theorem \ref{Thm5}\ and Lemma \ref{lemma8}.

\begin{theorem}
\label{Thm7}The process $W\left(  t\right)  $\ has the following properties:

\noindent(i) $W\left(  0\right)  =0$ almost surely;

\noindent(ii) the process $W\left(  t\right)  $\ is Gaussian with mean zero;

\noindent(iii)
\[
E\left[  W(t)W(s)\right]  =\left\{
\begin{array}
[c]{ll}%
\min\left(  \left\Vert t\right\Vert _{p},\left\Vert s\right\Vert _{p}\right)
& \text{if }t\neq0\text{ and }s\neq0\\
0 & \text{if }t=0\text{ or }s=0;
\end{array}
\right.
\]

\noindent(iv) let $\ t_{1},t_{2},t_{3},t_{4}$ in \textit{ }$\mathbb{Q}_{p}%
^{N}\ $\ such that $\left\Vert t_{1}\right\Vert _{p}\leq\left\Vert
t_{2}\right\Vert _{p}<\left\Vert t_{3}\right\Vert _{p}\leq\left\Vert
t_{4}\right\Vert _{p}$, then $W\left(  t_{2}\right)  -W\left(  t_{1}\right)  $
and $W\left(  t_{4}\right)  -W\left(  t_{3}\right)  $ are independent.
\end{theorem}

Bikulov and Volovich \cite{Bi-Volovich}, and Kamizono \cite{Kamizono},
constructed Brownian motion with $p$-adic time. In the case $N=1$, our
covariance function does not agree with the one given in \cite{Bi-Volovich},
\cite{Kamizono}, thus, in this case our result gives a different stochastic process.

\begin{acknowledgement}
The author wishes to thank to Sergii Torba and the anonymous referees for many
useful comments and discussions, which led to an improvement of this work.
\bigskip
\end{acknowledgement}


\begin{thebibliography}{99}                                                                                               %


\bibitem {Abd et al}Abdesselam Abdelmalek, Chandra Ajay, Guadagni Gianluca,
Rigorous quantum field theory functional integrals over the $p$-adics I:
anomalous dimensions, arXiv:1302.5971.

\bibitem {A-K-S}Albeverio S., Khrennikov A. Yu., Shelkovich V. M., Theory of
$p$-adic distributions: linear and nonlinear models. Cambridge University
Press, 2010.

\bibitem {Albeverio-et-al-1}Albeverio Sergio, Gottschalk Hanno, Wu, Jiang-Lun,
SPDEs leading to local, relativistic quantum vector fields with indefinite
metric and nontrivial S-matrix. Stochastic partial differential equations and
applications (Trento, 2002), 21--38, Lecture Notes in Pure and Appl. Math.,
227, Dekker, New York, 2002.

\bibitem {Albeverio-et-al-2}Albeverio Sergio, Rüdiger Barbara, Wu, Jiang-Lun,
Analytic and probabilistic aspects of Lévy processes and fields in quantum
theory. Lévy processes, 187--224, Birkhäuser Boston, Boston, MA, 2001.

\bibitem {Albeverio-et-al-3a}Albeverio S., Gottschalk H., Wu J.-L., Euclidean
random fields, pseudodifferential operators, and Wightman functions.
Stochastic analysis and applications (Powys, 1995), 20--37, World Sci. Publ.,
River Edge, NJ, 1996.

\bibitem {Albeverio-et-al-3}Albeverio Sergio, Gottschalk Hanno, Wu Jiang-Lun,
Convoluted generalized white noise, Schwinger functions and their analytic
continuation to Wightman functions, Rev. Math. Phys. 8 (1996), no. 6, 763--817.

\bibitem {Albeverio-et-al-4}Albeverio Sergio, Wu Jiang Lun, Euclidean random
fields obtained by convolution from generalized white noise, J. Math. Phys. 36
(1995), no. 10, 5217--5245.

\bibitem {Alberio-Kozyrev}Albeverio Sergio, Kozyrev Sergei V, Multidimensional
basis of $p$-adic wavelets and representation theory, $p$-Adic Numbers
Ultrametric Anal. Appl. 1 (2009), no. 3, 181--189.

\bibitem {Berg-Forst}Berg Christian, Forst Gunnar, Potential theory on locally
compact abelian groups. Springer-Verlag, New York-Heidelberg, 1975.

\bibitem {Bre-fre-ols-wi}Brekke L., Freund P. G. O., Olson M., Witten E.,
Nonarchimedean string dynamics,\ Nucl. Phys. B 302 (3), 365--402 (1988).

\bibitem {Bruhat}Bruhat François, Distributions sur un groupe localement
compact et applications à l'étude des représentations des groupes $p$-adiques,
Bull. Soc. Math. France 89 1961 43--75.

\bibitem {Bikulov}Bikulov A. Kh., Stochastic equations of mathematical physics
over the field of $p$-adic numbers, Theoret. and Math. Phys. 119 (1999), no.
2, 594--604.

\bibitem {Bi-Volovich}Bikulov A. Kh., Volovich I. V., $p$-adic Brownian
motion, Izv. Math. 61 (1997), no. 3, 537--552.

\bibitem {D-K-K-V}Dragovich, B.; Khrennikov, A. Yu.; Kozyrev, S. V.; Volovich,
I. V., On $p$-adic mathematical physics, $p$-Adic Numbers Ultrametric Anal.
Appl. 1 (2009), no. 1, 1--17.

\bibitem {Evans}Evans Steven N., $p$-adic white noise, chaos expansions, and
stochastic integration. Probability measures on groups and related structures,
XI (Oberwolfach, 1994), 102--115, World Sci. Publ., River Edge, NJ, 1995.

\bibitem {Fran}Frampton P. H., Okada Y., Effective scalar field theory of
$p$-adic string,\ Phys. Rev. D 37, 3077--3084 (1988).

\bibitem {Frenkel}Frenkel E., Langlands correspondence for loop groups.
Cambridge Studies in Advanced Mathematics, 103. Cambridge University Press,
Cambridge, 2007.

\bibitem {Gaisgory}Gaitsgory D., Notes on 2D conformal field theory and string
theory. In Quantum Fields and Strings: a Course for Mathematicians, Vol. 2
(Princeton, NJ, 1996/1997), Edited by P. Deligne et al. pp. 1017-1089,
American Math. Soc., Providence, RI, 1999.

\bibitem {G-Zuniga}Galeano-Peñaloza J., Zúñiga-Galindo W. A.,
Pseudo-differential operators with semi-quasielliptic symbols over p-adic
fields, J. Math. Anal. Appl. 386 (2012), no. 1, 32--49.

\bibitem {Gel-Vil}Gel'fand I. M., Vilenkin N. Ya, Generalized functions. Vol.
4. Applications of harmonic analysis. Academic Press, New York-London, 1964.

\bibitem {Glimm-Jaffe}Glimm James, Jaffe Arthur, Quantum physics. A functional
integral point of view. Second edition. Springer-Verlag, New York, 1987.

\bibitem {G-S}Ghoshal D., Sen A. , Tachyon condensation and brane descent
relations in $p$-adic string theory, Nucl. Phys. B 584, 300--312 (2000).

\bibitem {Hida et al}Hida Takeyuki, Kuo Hui-Hsiung, Potthoff Jürgen, Streit,
Ludwig, White noise. An infinite-dimensional calculus. Kluwer Academic
Publishers Group, Dordrecht, 1993.

\bibitem {Kamizono}Kamizono K., $p$-adic Brownian motion over $\mathbb{Q}_{p}%
$, Proc. Steklov Inst. Math. 265 (2009), no. 1, 115--130.

\bibitem {Kas-Wi}Kapustin A., Witten E., Electric-magnetic duality and the
geometric Langlands program, Commun. Number Theory Phys. 1 (2007), no. 1, 1-236.

\bibitem {Koch-Sait}Kochubei Anatoly N., Sait-Ametov Mustafa R., Interaction
measures on the space of distributions over the field of $p$-adic numbers,
Infin. Dimens. Anal. Quantum Probab. Relat. Top. 6 (2003), no. 3, 389--411.

\bibitem {Koch}Kochubei Anatoly N., Pseudo-differential equations and
stochastics over non-Archimedean fields. Marcel Dekker, Inc., New York, 2001.

\bibitem {kh0}Khrennikov A., Non-Archimedean analysis: quantum paradoxes,
dynamical systems and biological models. Kluwer, Dordreht, 1997.

\bibitem {kh1}Khrennikov A. Yu, Generalized functions and Gaussian path
integrals over non-Archimedean function spaces, USSR-Izv. 39 (1992), no. 1, 761--794.

\bibitem {kh2}Khrennikov A. Yu, Generalized functions on the non-Archimedean
superspace, USSR-Izvestiya 39 (1992), no. 3, 1209-1238.

\bibitem {K-Kos1}Khrennikov A. Yu., Kozyrev S. V., Ultrametric random field,
Infin. Dimens. Anal. Quantum Probab. Relat. Top. 9 (2006), no. 2, 199--213.

\bibitem {K-Kos2}Khrennikov A. Yu., Kozyrev S. V., Oleschko K., Jaramillo A.
G., Correa López, M., Application of $p$-adic analysis to time series, Infin.
Dimens. Anal. Quantum Probab. Relat. Top. 16 (2013), no. 4, 1350030 (15 pages).

\bibitem {K-Zh-1}Khrennikov A., Zhiyuan Huang, $p$-adic valued white noise,
Quantum Probability and Related Topics 9 (1994) 273-294.

\bibitem {K-Zh-2}Khrennikov A., Zhiyuan Huang, Generalized functionals of
$p$-adic white noise, Dokl. Akad. Nauk, 344, no. 1 (1995) 23-26.

\bibitem {Kolmogorov}Kolmogorov A. N. , The Wiener spiral and some other
interesting curves in Hilbert space, Dokl. Akad. Nauk S.S.S.R. 26 (1940), no.
2, 115-118.

\bibitem {Lax-Phillips}Lax Peter D., Phillips Ralph S., Scattering theory for
automorphic functions. Annals of Mathematics Studies, No. 87. Princeton Univ.
Press, Princeton, N.J., 1976.

\bibitem {Minlos}Minlos, R. A., Generalized random processes and their
extension to a measure. 1963 Selected Transl. Math. Statist. and Prob., Vol. 3
pp. 291--313 Amer. Math. Soc., Providence, R.I.

\bibitem {Mis}Missarov M. D., Random fields on the adele ring and Wilson's
renormalization group, Annales de l'institut Henri Poincare (A): Physique
Theorique 50 (3), 357-- 367 (1989).

\bibitem {Mis1}Missarov M.D., $p$-adic $\varphi^{4}$-theory as a functional
equation problem, Lett. Math. Phys. 39, No.3, 253-260 (1997).

\bibitem {Mis2}Missarov M.D., $p$-adic renormalization group solutions and the
Euclidean renormalization group conjectures, p-Adic Numbers Ultrametric Anal.
Appl. 4, No. 2, 109-114 (2012).

\bibitem {M-Z}Moeller N., Zwiebach B., Dynamics with in finitely many time
derivatives and rolling tachyons, JHEP 10, 034 (2002).

\bibitem {Obata}Obata Nobuaki, White noise calculus and Fock space. Lecture
Notes in Mathematics, 1577. Springer-Verlag, Berlin, 1994.

\bibitem {Palov-Fadeev}Pavlov B. S., Faddeev L. D., Scattering theory and
automorphic functions. (Russian) Boundary value problems of mathematical
physics and related questions in the theory of functions, 6. Zap. Nau\v{c}n.
Sem. Leningrad. Otdel. Mat. Inst. Steklov (LOMI) 27 (1972), 161--193.

\bibitem {R-Zuniga}Rodríguez-Vega J. J., Zúñiga-Galindo W. A., Elliptic
pseudodifferential equations and Sobolev spaces over p-adic fields, Pacific J.
Math. 246 (2010), no. 2, 407--420.

\bibitem {Su-Qiu}Su Weiyi, Qiu Hua, $p$-Adic Calculus and its Applications to
Fractal Analysis

and Medical Science, Facta Universitatis (Ni\v{s}): Ser. Elec. Energ. 21 (3),
339--347 (2008).

\bibitem {Smir}Smirnov V. A., Calculation of general $p$-adic Feynman
amplitude, Comm. Math. Phys. 149 (3), 623--636 (1992).

\bibitem {Taibleson}Taibleson M. H., Fourier analysis on local fields,
Princeton University Press, 1975.

\bibitem {Treves}Trèves François, Topological vector spaces, distributions and
kernels. Academic Press, New York-London 1967.

\bibitem {Var1}Varadarajan, V. S., Non-Archimedean models for space-time,
Modern Phys. Lett. A 16 (2001), no. 4-6, 387--395.

\bibitem {Var2}Varadarajan V. S., Reflections on quanta, symmetries, and
supersymmetries. Springer, New York, 2011.

\bibitem {Vlad1}Vladimirov V. S., On the non-linear equation of a $p$-adic
open string for a scalar field, Russ. Math. Surv. 60, 1077--1092 (2005).

\bibitem {Vlad2}Vladimirov V. S. , On the equations for $p$-adic closed and
open strings, $p$-Adic Numbers, Ultrametric Analysis and Applications1 (1),
79--87 (2009).

\bibitem {Vlad-Vol}Vladimirov V. S., Volovich I. V., $p$-adic quantum
mechanics, Comm. Math. Phys. 123 (1989), no. 4, 659--676.

\bibitem {Vlad-Vol2}Vladimirov V. S., Volovich Ya. I., Nonlinear dynamics
equation in $p$-adic string theory, Theor. Math. Phys. 138, 297--307 (2004); arXiv:math-ph/0306018.

\bibitem {V-V-Z-1}Vladimirov V. S., Volovich I. V., Zelenov E. I., Spectral
theory in $p$-adic quantum mechanics and representation theory, Soviet Math.
Dokl. 41 (1990), no. 1, 40--44.

\bibitem {V-V-Z}Vladimirov V. S., Volovich I. V., Zelenov E. I.: $p$-adic
analysis and mathematical physics, World Scientific, 1994.

\bibitem {Volovich1}Volovich I. V., Number theory as the ultimate physical
theory, $p$-Adic Numbers Ultrametric Anal. Appl. 2 (2010), no. 1, 77--87. This
paper corresponds to the preprint CERN-TH.4781/87, Geneva, July 1987.

\bibitem {Volovich2}Volovich I. V., $p$-Adic string,\ Class. Quant. Grav. 4
(1987), no. 4, L83--L87.

\bibitem {We1}Weil André, Basic Number Theory, Springer-Verlag, 1967.

\bibitem {Wi-1}Witten E. , Quantum field theory and the Jones polynomial,
Comm. Math. Phys. 121 (1989), no. 3, 351-399.

\bibitem {Wi-2}Witten E., Monopoles and four-manifolds, Math. Res. Lett. 1
(1994), no. 6, 769-796.

\bibitem {Zel1}Zelenov E. I., A $p$-adic infinite-dimensional symplectic
group, Russian Acad. Sci. Izv. Math. 43 (1994), no. 3, 421--441.

\bibitem {Zel2}Zelenov E. I., Quantum approximation theorem, $p$-Adic Numbers
Ultrametric Anal. Appl. 1 (2009), no. 1, 88--90.

\bibitem {Zu1}Zúñiga-Galindo W. A., Parabolic equations and Markov processes
over $p$-adic fields, Potential Anal. 28 (2008), no. 2, 185--200.

\bibitem {Zu2}Zúñiga-Galindo W. A., The Non-Archimedean Stochastic Heat
Equation Driven by Gaussian Noise, J. Fourier Anal. Appl. DOI 10.1007/s00041-014-9383-9.
\end{thebibliography}
\end{document}